\documentclass[reprint,floatfix,amsmath,amssymb,aps,pra,superscriptaddress]{revtex4-1}

\usepackage{graphicx}
\usepackage{dcolumn}
\usepackage{bm}
\usepackage[pdftex]{hyperref}
\usepackage{bbold}
\usepackage{dsfont}
\usepackage{amsthm} 
\usepackage{braket}
\usepackage{natbib}
\usepackage{float}
\usepackage{multirow}
\usepackage[table]{xcolor}
\usepackage{tikz}
\usepackage{comment}

\newcommand{\tr}{\text{tr}}

\def\lb{\left(}
\def\rb{\right)}
\newcommand{\cO}{\mathcal{O}}

\newcommand{\N}{\mathds{N}}
\newcommand{\ketbra}[2]{|#1\rangle\langle #2|}

\newcommand{\cS}{\mathcal{S}}
\newcommand{\cU}{\mathcal{U}}
\newcommand{\cT}{\mathcal{T}}
\newcommand{\cC}{\mathcal{C}}

\renewcommand{\tilde}[1]{\widetilde{#1}}
\renewcommand{\bar}[1]{\overline{#1}}
\newcommand{\cL}{\mathcal{L}}

\newcommand{\M}{\mathcal{M}}
\newcommand{\cH}{\mathcal{H}}

\newtheorem{theorem}{Theorem}
\newtheorem{cor}{Corollary}
\newtheorem*{lem*}{Lemma}

\newtheorem{lem}{Lemma}
\newtheorem{prop}{Proposition}

\begin{document}

\title{Limitations of optimization algorithms on noisy quantum devices}

\author{Daniel Stilck Fran\c{c}a}
 \email{dsfranca@math.ku.dk}
\affiliation{QMATH, Department of Mathematical Sciences, University of Copenhagen, Denmark}

 \author{Raul Garcia-Patron}
\email{rgarcia3@exseed.ed.ac.uk}
\affiliation{School of Informatics, University of Edinburgh, Edinburgh EH8 9AB, UK}


\date{\today}

\begin{abstract}
Recent technological developments have focused the interest of the quantum computing community on investigating how near-term devices could outperform classical computers for practical applications. A central question that remains open is whether their noise can be overcome or it fundamentally restricts any potential quantum advantage. We present a transparent way of comparing classical algorithms to quantum ones running on near-term quantum devices for a large family of problems that include optimization problems and approximations to the ground state energy of Hamiltonians. Our approach is based on the combination of entropic inequalities that determine how fast 
the quantum computation state converges to the fixed point of the noise model, together with 
established classical methods of Gibbs state sampling.
The approach is extremely versatile and allows for its application to a large variety of problems, noise models and quantum computing architectures. We use our result to provide estimates for a variety of problems and architectures that have been the focus of recent experiments, such as quantum annealers, variational quantum eigensolvers, and quantum approximate optimization. The bounds we obtain indicate that substantial quantum advantages are unlikely for classical optimization unless the current noise rates are decreased by orders of magnitude or the topology of the problem matches that of the device. This is the case even if the number of qubits increases substantially. 
We reach similar but less stringent conclusions for quantum Hamiltonian problems.
\end{abstract}

\maketitle
Quantum computation experiments have reached the classical complexity frontier, where the devices are
hard or impossible to simulate with classical resources \cite{Arute2019}.
A question that has focused the attention of the quantum computation community in the last years
is whether these near-term quantum computing devices, despite their inherent noise and lack of quantum error-correction, have any potential to demonstrate superiority over classical 
computing devices for practically relevant problems \cite{Preskill_2018}. 

A growing body of work has suggested using near-term devices to solve both optimization problems relevant in a large portfolio of industrial applications 
and approximating the ground state energy of physically relevant Hamiltonian~\cite{Peruzzo_2014,Kokail_2019,PhysRevA.99.062304,LaRose_2019,bravoprieto2019variational,McArdle_2019}. 
As opposed to quantum advantage tests based on sampling of random circuits~\cite{Arute2019}, optimization has a clear practical application and at least one straightforward certification protocol of its superiority over a classical computer: 
whichever device outputs the state with the lowest value of the cost function or Hamiltonian energy wins. 
This makes them ideal candidates for a convincing practical quantum advantage demonstration.  
The central question that remains open is whether a quantum advantage for practical problems is at all possible, especially taking into consideration the unavoidable presence of noise. 

In this article, we develop a new technique to provide sufficient conditions for when a noisy quantum device cannot significantly outperform a classical computer when it comes to optimizing a cost function given a noise model.
Additionally, we develop a technique to certify that a classical algorithm outperforms a given noisy quantum device 
running for a certain depth or number of cycles. This is shown using a nontrivial lower bound on the achievable value of the cost function 
or energy by the noisy quantum device. 
Both techniques can be adapted to a plethora of problems, noise models and computing architectures, opening the room for analytical proofs and concrete analysis of the limitations of near-term quantum computing devices. As a demonstration of its potential, we apply them to discuss the effect of local depolarizing noise on quantum circuits used to solve classical Ising models, which are equivalent to well-know NP-complete optimization problems, and variational quantum eigensolvers for local Hamiltonians problems \cite{moll2017quantum, kandala2017hardware, wang2019accelerated,farhi2014quantum}. We also generalize the technique to include noise models beyond depolarizing noise and develop a continuous time version, which allows us to study the effect of noise on quantum annealers \cite{farhi2000quantum,RevModPhys.90.015002}. 

Our analysis suggests that, for classical optimization problems, 
speedups may be impossible with current noise rates and are unlikely even at the boundary of fault-tolerant errors rates,
especially when the architecture does not match the topology of the problem. 
The problems where near-term devices could have an opportunity window are most probably quantum many-body physics motivated, where the architecture topology has been tailored-made to match the Hamiltonian of the problem and the correlation length in the system is bounded. Whether examples of this kind exist and cannot be simulated efficiently classically is, up to our knowledge, an important and non-trivial 
open problem. Furthermore, our bounds show that it is unlikely that increasing the number of qubits of current quantum devices without reducing the noise of the gates will lead to quantum speedups.


\section{An introduction to the technique}

In this work we are interested in solving optimization problems that can be 
recast as the minimization of a cost function $\tr\lb\rho H\rb$, where $\rho$ is
a quantum state of $n$ qubits, which therefore includes sub-problems defined over probability distributions on bit strings, and $H$ is a Hermitian operator characterizing the cost function of the problem. Many well-studied optimization problems, such as MAXCUT, 
can be recast as optimizing classical Ising Hamiltonians~\cite{Lucas_2014}, a connection exploited in most of the proposal of quantum algorithms for those problems and behind the architecture design of quantum annealers. In the case of a many-body ground-state problem, $H$ would correspond to the Hamiltonian of the system \cite{moll2017quantum}. 

Every Hamiltonian $H$
has an associated family of Gibbs states, $\sigma_\beta=e^{-\beta H}/\mathcal{Z}_\beta$,
where $\mathcal{Z}_\beta=\tr\lb e^{-H}\rb$ is the partition function and $\beta$ is the inverse temperature parameter, i.e, $\sigma_\beta$ corresponds to the thermal equilibrium state with temperature
$1/\beta$. The infinite temperature case $\beta=0$ corresponds to a fully mixed state of $n$ qubits,
the case $\beta=\infty$ provides the ground-state of $H$, which for optimization problems has its support over the set of solutions.

In what follows we present an approach to construct a sufficient condition that guarantees that
a noisy quantum computation solving a given optimization problem does not perform better
than its classical counterpart. 
As sketched in Figure \ref{fig:sketch}, the techniques combines relative entropy convergence methods (subsection \ref{subsec:noiseconv}), 
mirror descent (subsection \ref{subsec:mirrordescent}) and results on efficient algorithms for Gibbs state sampling.
The Gibbs sampling result being dependent on the specific problem the quantum algorithm is solving, we delay its presentation to the
sections where specific examples are presented (sections \ref{sec:IsingH} and \ref{sec:VQE}). In addition, we construct in 
subsection \ref{subsec:lowerbound} a technique that allows to certify that a classical algorithm does better than a given noisy device.

\begin{figure}[h!]
\centering
\includegraphics[width=1\columnwidth]{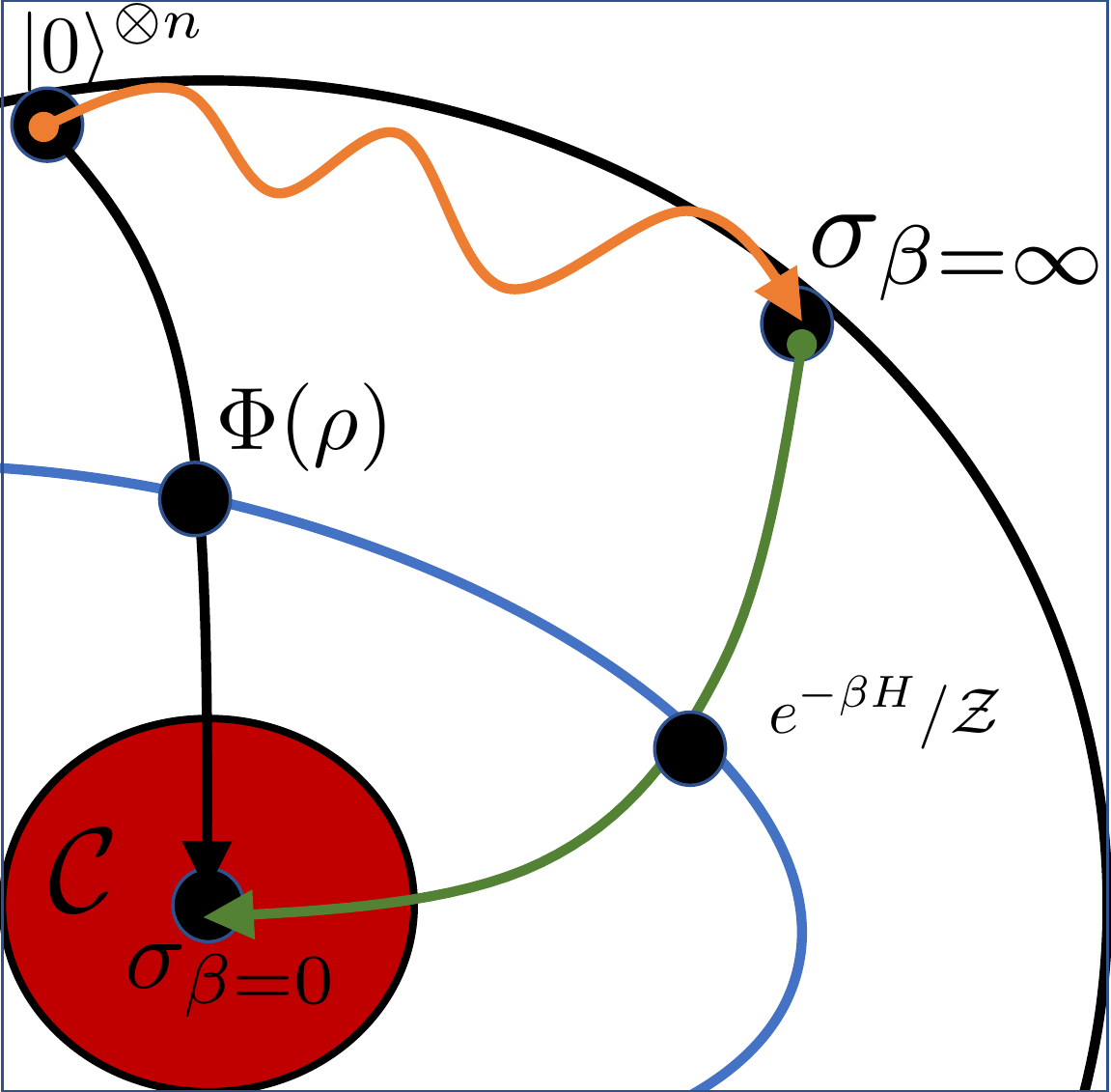}
\caption{The noisy quantum computer initialized at $|0\rangle^{\otimes n}$ attempts to follow a path (orange arrow) 
over the hyper-surface of the hyper-Bloch-sphere of $n$ qubits until reaching the targeted solution
of the problem with Hamiltonian $H$, $\sigma_{\beta=\infty}$, which corresponds to the Gibbs state at zero temperature. The noise imposes an alternative path (black arrow) that directs the computation toward the fixed point of the noise, the maximally mixed state (Gibbs state with $\beta=0$). Mirror descent allows us to associate to every intermediate state of the computation
$\Phi(\rho)$ an equivalent Gibbs state $e^{-\beta H}/\mathcal{Z}$ that provides the same cost function estimate as $\Phi(\rho)$ up to 
relative error $\epsilon$. Therefore, there is a one-to-one correspondence between the noisy computation path and an equivalent Gibbs state path (green arrow). More precisely, there is a manifold of states with approximately the same energy as $e^{-\beta H}/\mathcal{Z}$ (blue line). For many classes of problems, for any $\beta\leq\beta_c$ there will exist an efficient classical algorithm that samples from the Gibbs state, demarcating a region of efficient classical simulation (red circle). Relative entropy convergence methods can then be used to certify that the noisy quantum computer has entered the classically efficient region.}

\label{fig:sketch}
\end{figure}

\subsection{The effect of noise}
\label{subsec:noiseconv}

We will assume that we wish to implement a quantum circuit on $n$ qubits consisting of $D$ layers of unitaries $\cU_1,\ldots,\cU_D$ acting on $n$ qubits.
However, due to imperfections in the implementation, they are interspersed by a quantum channel $T$ that describes the noise, which for simplicity we will assume to be local and uniform over the circuit and time-invariant. The output of the noisy quantum circuit is then described by the quantum channel $\Phi=\bigcirc_{t=1}^D T\circ \mathcal{U}_t$. 
We consider noise models that drive the system to a fixed point $\sigma$ s.t. $T(\sigma)=\sigma$, which is the generic case. 
For our framework, it will be important to quantify how fast this happens when measured in the relative entropy $D(\rho\|\sigma)$. The following lemma exemplifies this for local depolarizing noise.
\begin{lem}\label{thm:concentration}
For $T$ being a layer of $1-$qubit depolarizing channels on $n$ qubits with depolarizing probability $p$, the channel  $\Phi=\bigcirc_{t=1}^D T\circ \mathcal{U}_t$ satisfies
\begin{align}\label{equ:strong-data-proc}
    D(\Phi(\rho)||\sigma)\leq (1-p)^{2D}D(\rho||\sigma)\leq (1-p)^{2D}n
\end{align}
for all states $\rho$ \cite{M_ller_Hermes_2016}. 
\end{lem}

\subsection{An alternative path}
\label{subsec:mirrordescent}

The key tool to achieving our goal is to decouple our result from the simulation of the noisy quantum circuit. We achieve this with mirror descent~\cite{Tsuda2005,Bubeck2015}, which was recently used in several results in quantum computing~\cite{1802.09025,Aaronson_2018,Brandao2017b,fern2019faster,Youssry_2019}, and the variational formulation of the relative entropy~\cite{Petz_1988}.
\begin{lem}\label{thm:mirrordescent}
Let $\Phi(\rho)$ be the output of a noisy quantum device, $\sigma=I/2^n$,  
and the relative entropy $D(\Phi(\rho)||\sigma)=n-S(\Phi(\rho))$, where $S(\Phi(\rho))=-\tr(\lb \Phi(\rho)\lb \log(\Phi(\rho)\rb\rb$. 
Then for any Hamiltonian $H$ there is a $\lambda\in [0,D(\Phi(\rho)||\sigma)]$ such that the 
state 
\begin{align}\label{equ:thermalstateclassical}
    \tilde{\sigma}=\text{exp}\lb -\frac{\lambda}{\|H\|\epsilon}H\rb/\mathcal{Z}
\end{align}
satisfies:
\begin{align}\label{equ:goodenergy2}
    \tr\lb H\lb \tilde{\sigma}-\Phi(\rho)\rb\rb\leq \|H\|\epsilon.
\end{align}
\end{lem}
This means that for every (noisy) quantum circuit output $\Phi(\rho)$, solving the optimization problem with cost function $H$, one can assign a Gibbs state of temperature
\begin{align}\label{equ:goodhamiltonian}
    \beta=\frac{\lambda}{\|H\|\epsilon}
\end{align}
which provides a solution that approximates with $\epsilon\|H\|$ accuracy that of the quantum circuit. Note that typically an absolute error of $\epsilon\|H\|$ reflects a relative error of $\epsilon$ of the energy.
Results like Lemma~\ref{thm:concentration} allow us to bound $\lambda$ given the depth. 

In subsection ~\ref{subsubsec:epsilon}
we argue why we should take $\epsilon$ to be approximately $10^{-1}-10^{-2}$ for the results we present here,
and in Sec.~\ref{sec:GeneralCircuits} we generalize to noisy channels with fixed points other than maximally mixed qubits. It is important to stress that our analysis is fundamentally different from analysing how close the output state is to the maximally mixed in trace distance, as discussed in more detail in Sec.~\ref{subsubsec:logNscaling} of the supplemental material. 

\subsection{Certifying classical superiority}
\label{subsec:lowerbound}
Exploiting the variational formulation of the relative entropy it possible to 
lower-bound the energy of the output of the noisy quantum device:
\begin{prop}\label{thm:lowerbound}
Let $\Phi(\rho)$ be the output of a noisy quantum device and $\sigma=I/2^n$ the fixed point of the noise. Then $\Phi(\rho)$ satisfies:
\begin{align}\label{equ:lowerboundenergy}
&\tr\lb \Phi(\rho) H\rb\nonumber\\&\geq\sup\limits_{\beta>0}\beta^{-1}\lb n-\log(\tr\lb e^{-\beta H}\rb)-D(\Phi(\rho)||\sigma)\rb.
\end{align}
\end{prop}
Thus, if we can approximate the partition function for some values of $\beta$ and bound the relative entropy with a relative entropy convergence method, we immediately obtain a lower bound on the achievable energy of the circuit through~\eqref{equ:lowerboundenergy}, which can be 
used to certify that the output of a given classical method outperforms the noisy quantum device 
solving the same problem over a given number of computation cycles $D$.
In Sec.~\ref{sec:variational} of the supplemental material we present a detailed proof of this proposition and its generalization beyond the maximally mixed state, while also discussing the scaling of the bound in more detail. But the best way of obtaining an intuitive grasp of its consequences is to perform a Taylor expansion of $f(\beta)=\log(\tr\lb e^{-\beta H}\rb)$ around $\beta=0$. Assuming that $f$ is analytical for a neighbourhood around $0$, we see that the bound in Eq.~\eqref{equ:lowerboundenergy} reads that for any $\beta$
\begin{align*}
    \tr\lb \Phi(\rho) H\rb\geq \tr\lb \sigma H\rb-\cO(\beta f^{(2)}(0))- \beta^{-1}D(\Phi(\rho)||\sigma),
\end{align*}
where $f^{(2)}(0)=2^{-n}\tr\lb H^2\rb-2^{-2n}\tr\lb H\rb^2$. 
One can readily check that for $\kappa$-local Hamiltonians we have $f^{(2)}(0)=\cO(\kappa n)$. We conclude that if $D(\Phi(\rho)||\sigma)\leq \epsilon^{2}n$, then picking $\beta=\epsilon$ yields $\tr\lb \Phi(\rho) H\rb\geq\tr\lb \sigma H\rb-\cO(\kappa\epsilon n)$.  That is, the output energy density will be close to the one of the fixed point.


\section{Classical Ising Hamiltonian problems}
\label{sec:IsingH}

By mapping NP-complete optimization problems such as MAXCUT to Ising Hamiltonian problems~\cite{Lucas_2014}, we aim at
finding the ground state and energy of the Ising Hamiltonian 
\begin{align}\label{equ:isingmodel}
    H_I=-\sum\limits_{i\sim j}a_{i,j}Z_iZ_j-\sum\limits_{i}b_iZ_i.
\end{align}
defined on  a graph $G=(V,E)$ with maximum degree $\Delta$ with $n$ vertices, where $i\sim j$ stands for the edge of graph $G$
between vertices $i$ and $j$ and $a_{i,j},b_i$ real. 
Monte Carlo algorithms to sample from $\sigma_\beta$ for $\Delta$-regular graphs have been extensively studied, see e.g.~\cite[Chapter 15.1]{levin2017markov}. 
There is a vast literature dedicated to determining for which range of $\beta$ it is possible to sample efficiently depending on the structure of the coupling coefficients and graph. To the best of our knowledge, the state of the art general condition that asserts rapid mixing for Ising models is the recent~\cite[Theorem 11]{2007.08200}, which reads that as long as 
\begin{align}\label{equ:conditionfastmixing1}
   \beta\|A\|<1,
\end{align}
where $(A)_{i,j}=a_{i,j}$ is the $n\times n$ coupling matrix, it is possible to approximately sample from the Gibbs state $\sigma_\beta$ in polynomial time. For instance, if we have a $\Delta$-regular graph and $|a_{i,j}|\leq1$, note that $\|A\|\leq\Delta$ and the bound above yields $\beta<\Delta^{-1}$. On the other hand, for the SK model, where we have a complete interaction  graph and $a_{i,j}\sim\mathcal{N}(0,n^{-1})$,~\cite{2007.08200} shows that rapid mixing still holds at $\beta<1/4$. 

\subsection{Bounding the number of cycles}\label{sec:numberofcycles}
Let us now bound the maximum depth for MAXCUT on a noisy computer and $\Delta$-regular graph. 
Let's consider a circuit of depth $D$, where $f_1$ is the fraction of layers of single qubits gates with error $p_1$, $f_2$ the fraction of layer of two-qubit gate of error $p_2$, and $p_m$ the probability of error of the measurements.  Replacing every ideal component by its noisy counterpart followed by local depolarizing noise with rate $p_i$, combining Lemma \ref{thm:mirrordescent}, eq.~\eqref{equ:goodhamiltonian}, the bound eq.~\eqref{equ:conditionfastmixing1}, and the approximation $-\log(1-p)\simeq p^{-1}$ for small $p$ we obtain:
\begin{prop}\label{prop:ising}
Let $H_I$ be an Ising Hamiltonian as in Eq.~\eqref{equ:isingmodel}. Moreover, let $\Phi(\rho)$ be the output of a noisy quantum circuit on $n$ qubits consisting of $D$ layers of unitaries. Suppose further that the circuit consists of a fraction of $f_1$ layers of one-qubit gates and $f_2$ layers of $2$-qubit gates, affected by local depolarizing noise with probability $p_1,p_2$ respectively and the measurement is affected with depolarizing probability $p_m$. Then for $D\geq D_{\max}$ with
\begin{align}\label{equ:boundD}
    D_{\max}=\frac{\log\epsilon^{-1}+\log\lb \|H_I\|^{-1}\|A\|n\rb-p_m}{2(f_1p_1+f_2p_2)},
\end{align}
there is a Gibbs state $\tilde{\sigma}$ satisfying 
\begin{align*}
    \tr\lb H_I\lb \tilde{\sigma}-\Phi(\rho)\rb\rb\leq  \epsilon\|H\|.
\end{align*}
that can be sampled from in polynomial time.
\end{prop}
 For instance, for $\Delta$ regular graphs with $a_{i,j}\in\{0,1\}$ the energy of the ground state is of order $\Delta n$. As $\|A\|\leq \Delta$, up to a relative error of $\epsilon$ we can sample from a comparable state after at most $\log(\epsilon^{-1})/2(f_1p_1+f_2p_2)$ depth.

We have now everything at hand to make some estimates on the capability for
realistic architectures to solve practical problems. We will assume the values
$p_1=1.6\times10^{-3}$, $p_2=6.2\times10^{-3}$ and $p_m=3.8\times10^{-2}$ taken from 
the recent experiment \cite{Arute2019}. A simple calculation using eq.~(\ref{equ:boundD}) with $\epsilon=10^{-1}$ shows that those numbers limit the depth of any quantum circuit to $D_c\approx 300$ for circuits dominated by $2-$qubit gate layers and $D_c\approx 700$ for $1-$qubit gate dominated circuits. 

Most current quantum computing architectures have a planar design. However, usually
interesting optimization problems are defined over non-planar graphs, which imposes the need for
SWAP circuits that consume a lot of cycles of computation.
For example, the SK model studied in \cite{Arute2019} requires $7n$ layers cycles of computation per application of the QAOA unitary and implements 3 layers of QAOA. 
The bound in Eq.~\eqref{equ:boundD}, for which we have $\beta_c\geq1/4$, predicts a limit of roughly $20$ qubits before their systems lose any potential advantage at three rounds.
Indeed, experimental observations show that for systems of size 17 their device is not better than
random guessing. 

Furthermore, it is also possible to approximately compute the partition function for the aforementioned $\beta$ range. We evaluated the bound in Eq.~\eqref{equ:lowerboundenergy} for several instances of the SK model on up to $30$ qubits and consistently observed that already depths $D_c\approx 150$ are outperformed by Gibbs states in the polynomial-time range. Translating this bound to number of qubits, it predicts that our classical method should outperform the noisy quantum computer with more than $10$ qubits. In Sec.~\ref{sec:variationalexample} of the supplementary material we explain how we perform such computations in more detail and that the variational bound allows for comparison to any other algorithm to approximately solve optimization problems beyond Gibbs sampling. Also note that quantum computers of this size or smaller can be easily simulated on a laptop.

We remark that equation~\eqref{equ:boundD} has the advantage of providing a rigorous guarantee
for arbitrary system sizes and instances, where~\eqref{equ:lowerboundenergy} 
provides sharper estimates for specific instances of a problem. 
Thus, both methods seem to complement each other in providing a good bound on the range of 
cycles of computation a noisy quantum computers can sustain 
before it is outperformed by polynomial time classical methods.

Even for optimization problems on sparser graphs, such as the MAXCUT problems
also studied in \cite{Arute2019}, the depth of every QAOA layer is proportional 
to $n$. From our previous discussion, we readily infer that, as long as the depth of each QAOA layer scales with $n$, also for sparser problems the noise rate would have to decrease roughly two orders of magnitude, well below the fault-tolerant threshold, for quantum computers to have a chance of being competitive with classical methods. This is because current heuristic methods yield good results for graphs with numbers of nodes in the $10^3-10^4$ range~\cite{Kochenberger_2011,Dunning_2018} and state of the art SDP solvers can solve $10^4$-node instances in seconds on a laptop~\cite{yurtsever2019scalable}.


\section{Variational Quantum Eigensolvers}
\label{sec:VQE}
 
Variational quantum eigensolvers (VQE) are one of the most studied families of quantum algorithms for near-term quantum devices \cite{moll2017quantum, kandala2017hardware, wang2019accelerated}. This family of hybrid classical-quantum algorithms 
supplements a traditional variational technique on a classical computer with
a call to a quantum circuit that encodes the quantum state and measures its energy.
The technique allows to circumvent the cost of encoding and evolving quantum states in
a traditional classical computer, but suffers from similar limitation as any traditional
variational approach, among them the fact that we restrict the set of  solution to
a manifold of the Hilbert space. Here we exemplify how to derive limitations on VQE for finding the ground state of local Hamiltonians, although it is possible to extend the results here to other classes of problems.

There are several results and techniques for the classical approximation of quantum Gibbs states~\cite{2002.02232,1910.09071,Kuwahara_2020,PhysRevLett.93.207204,RevModPhys.73.33,PhysRevLett.97.187202,Tang_2013,PhysRevLett.102.190601,PhysRevB.91.045138}.
For instance, in the recent~\cite[Corollary 21]{1910.09071} the authors show that on a Hamiltonian on a $d$ dimensional lattice and $\kappa$ local interactions of strength at most $J$,  for all inverse temperatures  $\beta\leq \beta_c$ with 
\begin{align}
\beta_c \geq \lb 5e \kappa d^\kappa J\rb^{-1},
\end{align}
we can approximately compute the partition function of $\mathcal{Z}_\beta=e^{-\beta H}$ up to a multiplicative error $\delta$ in $\cO(n^{\log(n \delta^{-1})})$ time. Let us normalize $\|H\|$ in such a way that $\kappa J n\leq\|H\|\leq2\kappa Jn$. 
By our choice of normalization of $\|H\|$ we obtain the following proposition:
\begin{prop}\label{prop:depth_classical}
Let $H$ be a local Hamiltonian on a $d$ dimensional lattice with $\kappa$ local interactions of strength at most $J$. Moreover, let $\Phi(\rho)$ be the output of a noisy quantum circuit on $n$ qubits consisting of $D$ layers of unitaries interspersed with local depolarizing channels with depolarizing rate $p$. Then for $D\geq D_{\max}$ with
\begin{align}\label{equ:maximumdepth}
   D_{\max}= \frac{\log\lb 20 e\lb d^\kappa  \epsilon\rb^{-1}\rb}{2p}
\end{align}
there is a Gibbs state $\tilde{\sigma}$ satisfying 
\begin{align*}
    \tr\lb H\lb \tilde{\sigma}-\Phi(\rho)\rb\rb\leq 2\kappa Jn \epsilon.
\end{align*}
whose energy w.r.t. $H$ can be computed in quasi-polynomial time.
\end{prop}

This upper-bound on the maximum depth of a VQE algorithm implies also a limitation on the maximum correlation length $\xi$ of ground states we can prepare with a near-term device. As detailed in Section~\ref{sec:corlength} of the supplementary material, a light cone argument can be used to argue that whenever $D<\xi$ the quantum circuit will provide a solution that is far in trace distance from the actual ground state. 
Thus, the noise imposes limits on the correlation length present in
the quantum ansatz generated by near-term devices.

Comparing the bound in Eq.~\eqref{equ:boundD} and Eq.~\eqref{equ:maximumdepth}, we see that they exhibit the same scaling. Thus, we believe that this warrants similar conclusions: as long as the depth of a VQE layer scales with the system size, substantially decreasing the noise is more important than increasing system size. However, as quantum Hamiltonians are inherently more complex than Ising problems, and also less studied, 
smaller system sizes and higher noise rates might be tolerable for a quantum advantage demonstration.


\section{Noisy quantum annealers}

Quantum annealers have been proposed as an alternative architecture to the circuit model of computation 
to solve classical Ising problems~\cite{Johnson_2011}. 
They follow a quantum computation approach to optimization of Ising Hamiltonians aiming at 
preparing the ground state $\ket{\psi_I}$ of the classical Hamiltonian $H_I$ of eq.~\eqref{equ:isingmodel} adiabatically.
We start by first preparing the ground state of $H_0=- \sum\Gamma_i X_i$, which
is easily seen to be the state $\ket{\psi_0}=\ket{+}^{\otimes n}$. 
Then by adiabatically shifting from the Hamiltonian $H_0$ toward $H_I$, through the family
of transverse-field Ising Hamiltonians
\begin{align*}
    H_s=g_I(s)H_I+g_0(s)H_0
\end{align*}
where $g_0,g_I$ are smooth functions of $s$ s.t. $g_0(0)=A,g_I(0)=0$ and $g_0(T)=0,g_I(T)=B$, we converge to $\ket{\psi_I}$ by evolving under the
Hamiltonian $H_s$over time $T$.

The spectral gap of $H_s$ along the curve determines how slow the evolution has to be to guarantee a good solution. Furthermore, to the best of our knowledge, most rigorous theorems 
yield a runtime that scales at least linearly in the number of qubits and inverse quadratically on the gap~\cite{Jansen_2007}. There are, however, heuristic bounds that only scale quadratically inverse with the gap ~\cite{PhysRevLett.102.220401}.
We refer to~\cite{RevModPhys.90.015002} for a discussion of adiabatic quantum computation and these issues.

\subsection{A continuous time contraction bound}

Quantum adiabatic computation being a continuous time computational model, we need to
generalize Lemma 1 to this framework. 
Therefore, we assume that the evolution of the system at time $s$ is described by the time-dependent Lindbladian
\begin{align*}
    \cS_s(\rho)=-i[\rho,H_s]+\cL(\rho).
\end{align*}
Here $\cL$ is a purely dissipative term, which we assume to be time-independent and that is supposed to model the thermalization of the system.
The time evolution of the system at time $t$ is described by the quantum channel $\cT_t$ given by
\begin{align*}
    \cT_t=\lim\limits_{n\to \infty}\bigcirc_{l=1}^n e^{\frac{t}{n} \cS_{t_l}}.
\end{align*}
We will consider a model of local noise combining amplitude damping, dephasing and control error for which 
the fixed point of the Lindbladian $\cS_s$ is the product state $\sigma=\sigma_\gamma^{\otimes n}$, where
$\sigma_\gamma=(e^{\gamma}+e^{-\gamma})^{-1}e^{-\gamma Z}$.

The continuous time version of the relative entropy decay of Eq.~\eqref{equ:strong-data-proc} is that $\cL$ satisfies a \emph{modified logarithmic Sobolev inequality} (MLSI) with constant $\alpha>0$~\cite{Kastoryano_2013,Olkiewicz_1999}. This implies that $\cL$ satisfies
\begin{align}\label{equ:MLSI}
    D\left(e^{t\cL}(\rho)||\sigma\right)\leq e^{-\alpha t}D\left(\rho||\sigma\right).
\end{align}
for all $n-$qubit states $\rho$ and times $t\geq0$ and a state $\sigma>0$. 

Exploiting the generalization to an arbitrary fixed point of Lemma~\ref{thm:mirrordescent}
in Section \ref{sec:GeneralCircuits} and its adaptation to the continuous evolution setting
in Section \ref{sec:SMquann}, both in the supplemental material, we obtain:
\begin{theorem}
Let  $\cS_s(\rho)=-i[\rho,H_s]+\cL(\rho)$ be a time-dependent Linbladian  s.t. $\cL$ satisfies MLSI with constant $\alpha>0$. Then the evolution from time $0$ to $t$ satisfies $D(\mathcal{T}_t(\rho)||\sigma)\leq e^{-\alpha t}D\left(\rho||\sigma\right)+X(t,H_s,\sigma)$,
where
\begin{align}\label{equ:entroadiabatic}
X(t,H_s,\sigma)=\int\limits_{0}^t d\tau e^{-\alpha(t-\tau)}\|\sigma^{-\frac{1}{2}}[H_t,\sigma]\sigma^{-\frac{1}{2}}\|.
\end{align}
\end{theorem}

\subsection{Noise model}
For current implementations of annealers, amplitude damping, control errors and dephasing are the main sources of noise~\cite{Johnson_2011}.
As shown in the supplementary material in Sec.~\ref{sec:noisemodel}, assuming an amplitude damping rate of $r_1$,
a dephasing $r_2$ and $r_3$ as the standard deviation on the control error on the
parameter $b_i$ and $\Gamma_i$, one can show that the fixed point reads $\sigma=\sigma_\gamma^{\otimes n}$, where
\begin{align}\label{equ:fixedstate}
\sigma_\gamma=\begin{pmatrix}
\frac{r_1+r_3}{r_1+2r_3}& 0 \\
0 &\frac{r_3}{r_1+2r_3}
\end{pmatrix}.
\end{align}
Thus, we see that when amplitude damping dominates the control error, i.e, the regime $r_1\gg r_3$, the fixed point is essentially the $\ket{0}$ state, where for $r_3\gg r_1$ we have a state close to maximally mixed state. One can also obtain a MLSI for $\cL$ with $\alpha(r_1,r_3)=r_1+2r_3$~\cite[Theorem 19]{Beigi_2020}.

\subsection{Example: linear adiabatic path}

Let us exemplify how inequality~\eqref{equ:entroadiabatic} behaves for the path $H_s=\frac{s}{T}H_0+\left(1-\frac{t}{T}\right)H_1$ if we  evolve from the initial state $\rho=\ketbra{+}{+}^{\otimes n}$ under $\cS_s$ from $0$ to $T$.

Using the noise models given in Eq.~\eqref{equ:fixedstate}, we conclude that the relative entropy after evolving the system for time $T$ is bounded by
\begin{align}\label{equ:boundentropyising}
&D(\mathcal{T}_T(\ketbra{+}{+}^{\otimes n})||\sigma)\leq n f(\gamma,r,T,\bar{\Gamma})
\end{align}
where
\begin{align}\label{equ:locadep}
&f(\gamma,r,T,\bar{\Gamma})= e^{-r T}\log(2\cosh(\gamma))\nonumber\\& +\frac{2\sinh(\gamma)\lb 1-e^{-rT}r T-e^{-rT}\rb}{r^2T} \bar{\Gamma},
\end{align}
and $\gamma=\frac{1}{2}\log\lb 1+\frac{r_1}{r_3}\rb$ and $\bar{\Gamma}=n^{-1}\sum_i\Gamma_i$. We refer to Sec.~\ref{sec:computations} of the supplementary material for a detailed derivation.

Using those results we conclude that we need to simulate the classical Gibbs state
\begin{align*}
        \tilde{\sigma}\propto\text{exp}\lb -\gamma \sum\limits_{i=1}^nZ_i-\frac{4\lambda}{\|H_I\|\epsilon}H_I\rb
\end{align*}
for $\lambda \leq f(\gamma,r,T,\bar{\Gamma})n$ to approximate the energy of the output $\mathcal{T}_t(\ketbra{+}{+}^{\otimes n})$ of the noisy quantum computer up to $\epsilon H_I$. Note that $\tilde{\sigma}$ is just the Gibbs state of a classical Ising model.

Thus, the same mixing time of Ising used in section \ref{sec:IsingH} can be used here again. Note that the mixing time bounds in Eq.~\eqref{equ:conditionfastmixing1} is independent of the external field and that adding the term $\sum\limits_{i=1}^nb_i'Z_i$ corresponding to the fixed point amounts to changing the external field strength. Translating this bound to our setting, we see we can sample from $\tilde{\sigma}$ in polynomial time for
\begin{align}\label{equ:classicalrealm}
    &f(\gamma,r,T,\bar{\Gamma})\leq \frac{4\|A\|n}{\|H_I\|}\epsilon.
\end{align}
This puts stringent constraints on the potential of noisy adiabatic quantum computers to outperform classical computers optimizing classical Ising Hamiltonians. 
We exemplify the behaviour of the bound in Fig.~\ref{fig:boundannealer}.

\begin{figure}[h!]
\centering
\includegraphics[width=1\columnwidth]{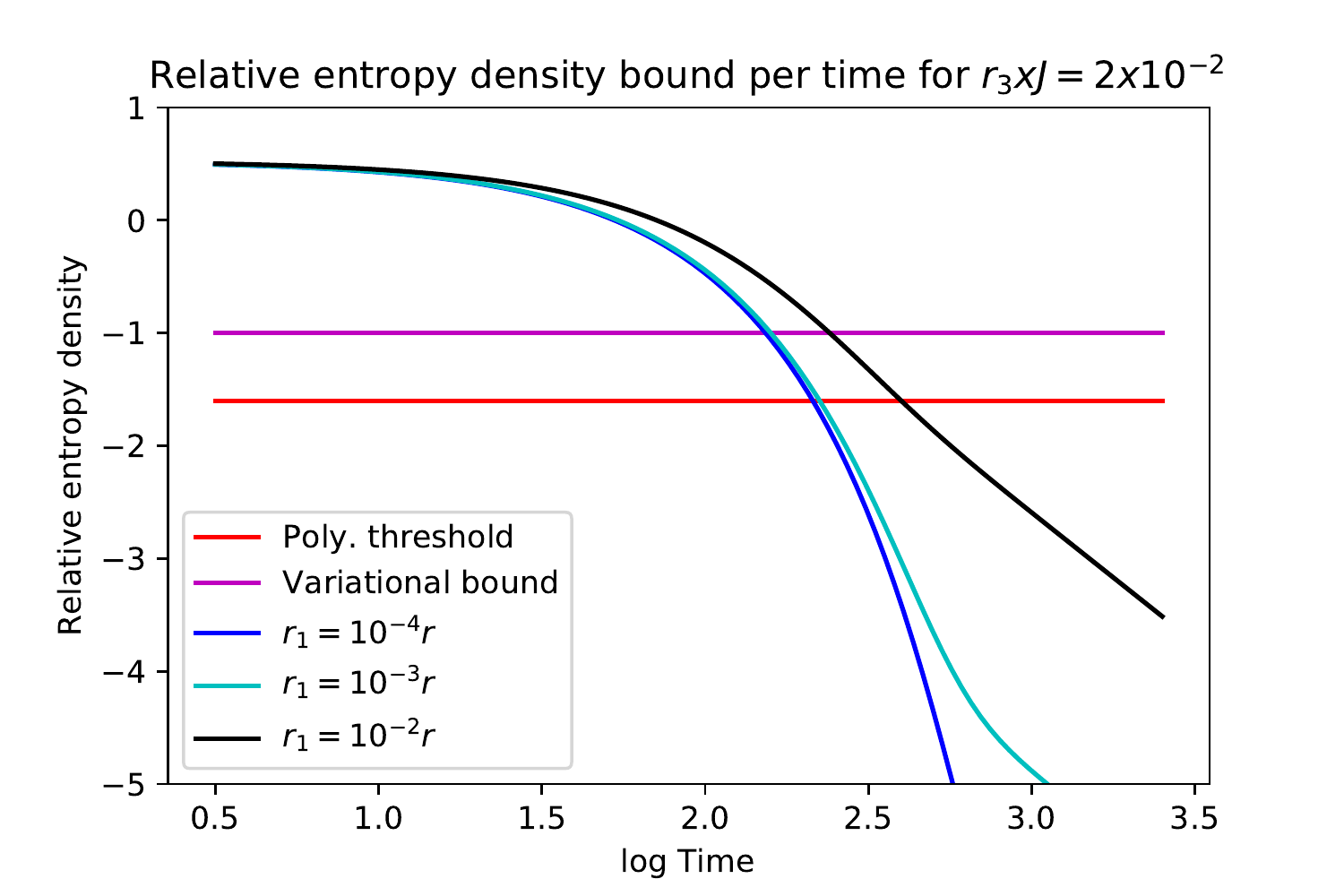}
\caption{Relative entropy density $D(\mathcal{T}_T(\ketbra{+}{+}^{\otimes n})||\sigma)/n$ as bounded by Eq.~\eqref{equ:boundentropyising} for a control error noise rate of $r_3\times J=2\times 10^{-2}$ . The polynomial threshold line depicts when the device provably does not significantly outperform polynomial methods and is based on Eq.~\eqref{equ:classicalrealm}. The variational bound is an extrapolation based on various instances of up to $30$ qubits and is based on Eq.~\eqref{equ:variationalentropy} of the supplemental material. The graph that models the interactions available in the device has maximum degree $\Delta=6$. 
The interaction strengths $a_{i,j},b_i,\Gamma_i$ of the device are unit-less and $\approx1$,
where the modulation functions satisfy $g_i(s)\in[10^{-2},10]$, measured in GHz. 
Specifications of the device \cite[Fig. 3.1, Fig. 3.2]{headquarters2019technical} 
suggest that the standard deviation of the control noise for $a_{i,j},b_i,\Gamma_i$
is not constant in time, but one can see that their product with $g_i(s)$ remains in a reasonable interval
of values always larger than $10$MHz. Therefore, we choose to re-scale the problem to 
$g_0(0)=g_1(T)=1$ and $r_3\times J\simeq 2\times 10^{-2}$, $J$ the maximum interaction strength. 
Unfortunately, the value of $r_1$ is not provided but~\cite{Albash_2015} suggests that $r_1/r_3\in[10^{-4},10^{-2}]$ should cover a realistic parameter range.}
\label{fig:boundannealer}
\end{figure}

Let us now discuss our bound given relevant parameters in current implementations. We will consider a commercially available quantum annealer with the parameters released in 2019~\cite{headquarters2019technical}. The results and technical details of the choice of parameters is shown in Fig.~\ref{fig:boundannealer}.
We conclude that the maximal annealing time for which the device is useful is of order $200\mu s$ and the window for a potential advantage lies on problems that only require a very small annealing time.
This is consistent with the recent findings of a potential advantage of commercial annealers versus the best classical solvers~\cite{Hauke_2020}. For instance, in the recent demonstration of a constant speedup of quantum annealers in~\cite{Mandr__2018}, the authors used the minimal annealing time available for their findings, $5\mu s$. The same is true for the results of~\cite{Denchev_2016}, which used $20\mu s$ annealing times. 
However, it is not clear to what extent this advantage does not result from purely engineering 
reasons disconnected to the presence of quantum effects, such as the high sampling rate on such devices. An argument that would further substantiate the claims of~\cite{Aramon_2019} that a new generation of special-purpose classical annealing chips have the potential to outperform current quantum annealing devices. 
Furthermore, we again conclude that increasing the number of qubits without substantially decreasing the control errors of annealers is unlikely to increase their potential to achieve an unquestionable quantum advantage.

\section{Conclusion and outlook}

In this work we have developed a framework to quantify the effect of noise on near-term quantum computers when performing optimization. It is based on the combination of three main ingredients: firstly, contractivity results that quantify how the noise makes the quantum system deviate from the ideal solution and gets closer to the fixed point of the noise process; secondly, the use of mirror descent that allows us to assign to every step of the quantum computation a Gibbs state achieving the same 
approximation of the cost function; thirdly, exploiting existing efficient algorithms to approximate the specific cost function of the problem at high noise levels. The techniques proved to be extremely versatile and applicable to a large variety of problems, noise models and quantum computing architectures.
Our work confirms that a bound on the depth that scales inversely proportional to the quantum gate error is indeed a universal bound and provides means of estimating the maximum depth rigorously.

Additionally, we have developed a technique to certify that a classical algorithms outperforms a given noisy quantum device running for a certain depth
or number of cycles. This is shown using a nontrivial lower bound on the achievable energies by the noisy quantum device. Interestingly, it can also be used to
provides sharper estimates for specific instances of a problem. 
Thus, a combination of both methods provides a good estimate of the limitations of a noisy near-term
quantum computer.

Our result suggests that there is little chance near-term quantum devices can provide an advantage over classical ones for classical optimization problems, especially if they do not match the architecture's topology. This is because this would result in depths scaling with system size for current algorithmic proposals, which in turn imply that noise rates would have to decrease two orders of magnitude before being potentially competitive against state of the art classical methods.

The problems where near-term devices could have an opportunity window are most probably quantum many-body physics motivated, where the architecture topology has been tailored to the Hamiltonian of the problem and the correlation length in the system is bounded. Whether examples exist that cannot be simulated efficiently classically is, up to our knowledge, an important open problem.

It is important to stress that our bounds guarantee the efficient classical simulation, but
at the price of being extremely conservative. Practical implementations of Gibbs samplers 
are expected to reach higher $\beta$ than predicted by our bounds~\cite{Isakov_2015},
especially using clusters or supercomputers.
Furthermore, it is possible that the relative entropy contracts significantly more for states of interest than predicted by our state-independent bound.
Moreover, relative entropy contraction techniques struggle when the fixed point of the system is close to a pure state, i.e. when amplitude damping is the main source of noise. This can be seen in our bounds for quantum annealers in the $r_1\simeq r_3$ regime for our quantum annealing bound.
Noise that does not uniformly contract the relative entropy, like dephasing, also requires a more detailed analysis, but technical tools are already available~\cite{bardet2018hypercontractivity}. Therefore, improvements to our technique
for both additional types of noises can only improve over the estimates obtained in this work. 

Noise mitigation techniques based on post-processing of the quantum computation measurement 
outcomes, despite being useful to filter the data from noise, would not change the predictions of our work. One can easily reach this conclusion through a simple argument using the data-processing inequality of relative entropy. It would be interesting to see if primitive error correction tools can be used to reduce the effect of noise, in the spirit of~\cite{Bremner_2017}, or how embedding a problem into the annealer topology or quantum annealing correction 
may be both limited by the technique presented here \cite{Job_2018}.

Finally, we believe that the technique developed here can be adapted to analog quantum simulators, 
where our continuous time result can be seen as a first step. This would allow to make predictions about the resilience of 
complex phases of matter to experimental imperfections. Similarly, extending our result 
to Fermionic problems is something we plan to address in future work.

DSF was supported by VILLUM FONDEN via the QMATH Centre of Excellence under Grant No. 10059. RGP was supported by the Quantum Computing and Simulation Hub, an EPSRC-funded project, part of the UK National Quantum Technologies Programme. We thank Heng Guo, Jonah Brown-Cohen and Pierre-Luc Dallaire-Demers for helpful discussions.
\bibliography{mybib}


\section*{Supplementary Information}

\section{Quantum circuits with depolarizing noise}\label{sec:relativentropy}

In this section we will prove and discuss our bounds in the special case of local depolarizing noise. Furthermore, we will review some fundamental properties of the relative entropy more generally and discuss some qualitative aspects of our results, like how they differ from analysing convergence in trace distance to the fixed point. Whenever the proof is essentially the same in the case of the maximally mixed state and more generally, like it is the case for Prop.~\ref{prop:variationalentropy}, we will prove it in full generality.

\subsection{Relative entropy fundamentals}
Most of our results rely on the relative entropy $D(\rho\|\sigma)$ between two states. In this part of the supplemental material  we collect some basic facts about it and then move on to a more detailed proof of Lemma~\ref{thm:mirrordescent}.
First, recall that for a full-rank state $\sigma$ and an arbitrary state $\rho$, the relative entropy $D(\rho\|\sigma)$ is defined as 
\begin{align*}
D(\rho\|\sigma)=\tr\lb \rho(\log(\rho)-\log(\sigma)\rb.
\end{align*}
One of the most important properties of the relative entropy is that it satisfies a data-processing inequality. That is, for any quantum channel $T$ and states $\rho,\sigma$ we have:
\begin{align*}
D(T(\rho)\|T(\sigma))\leq D(\rho\|\sigma).
\end{align*}
In this work we also exploited extensively the case in which the channel contracts the relative entropy strictly, which we call a strong data-processing inequality.

In the main text we exploited upper bounds on the relative entropy with a fixed second argument. The most straightforward way to derive them is through the max-relative entropy.
Indeed, one can show that the relative entropy is upper-bounded by the so-called max-relative entropy $D_\infty$~\cite{M_ller_Lennert_2013,Datta_2009}, which is defined as:
\begin{align*}
D_\infty(\rho\|\sigma)=\log(\|\sigma^{-\frac{1}{2}}\rho\sigma^{-\frac{1}{2}}\|),
\end{align*}
where $\|\cdot\|$ denotes the operator norm.
From this, we see that for all states $\rho$:
\begin{align*}
D(\rho\|\sigma)\leq \log(\|\sigma^{-1}\|).
\end{align*}
Indeed, we have
\begin{align*}
&D(\rho\|\sigma)\leq D_\infty(\rho\|\sigma)\\
& \leq \log(\|\rho\|\|\sigma^{-\frac{1}{2}}\|^2)\leq \log(\|\sigma^{-1}\|),
\end{align*}
where we used the submultiplicativity of the operator norm and the fact that $\|\rho\|\leq 1$ for all states.

We also have the following bound on the max relative entropy:
\begin{align*}
D_{\infty}(\mathcal{U}(\sigma)\|\sigma)\leq \log(\|\sigma^{-1}\|\|\sigma\|)
\end{align*}
for all unitary channels $\mathcal{U}$. This also easily follows from the submultiplicativity of the operator norm, as we have:
\begin{align*}
&D_{\infty}(\mathcal{U}(\sigma)\|\sigma)=\log(\|\sigma^{-\frac{1}{2}}\mathcal{U}(\sigma)\sigma^{-\frac{1}{2}}\|)\\
&\leq \log(\|\sigma^{-1}\|\|\mathcal{U}(\sigma)\|)=\log(\|\sigma^{-1}\|\|\sigma\|).
\end{align*}
If $\sigma$ is a Gibbs state $\sigma=e^{-\beta H}/\mathcal{Z}$, then the inequality above reads:
\begin{align*}
&D_{\infty}(\mathcal{U}(\sigma)\|\sigma)\leq \log( \|e^{-\beta H}\|\|e^{\beta H}\|)\\
&\leq 2\beta \|H\|.
\end{align*}
The last inequality follows from the fact that spectrum of $e^{-\beta H}$ and $e^{\beta H}$ are both contained in the interval $[e^{-\beta \|H\|},e^{\beta \|H\|}]$.

\subsubsection{From relative entropy convergence to trace distance}\label{subsubsec:logNscaling}
Let us now discuss the connection between bounds on the relative entropy and closeness in trace distance. This is established by Pinsker's inequality, which asserts that:
\begin{align*}
D(\rho||\sigma)\geq \frac{1}{2}\|\rho-\sigma\|_{\tr}^2.
\end{align*}
Thus, in particular, it follows from a strong data processing inequality for a channel $T$  with constant $\alpha$ that applying a channel $n$ times to any initial state is enough to ensure that the trace distance to $\sigma$ is at most:
\begin{align*}
\frac{1}{2}\|T^n(\rho)-\sigma\|_{\tr}^2\leq D(T^n(\rho)\|\sigma)\leq (1-\alpha)^n\log(\|\sigma^{-1}\|).
\end{align*}
Thus, if we apply the channel more than
\begin{align}\label{equ:mixingtime}
N_{\max}\geq \frac{\log(2\epsilon^2\log(\|\sigma^{-1}\|))}{\log((1-\alpha)^{-1})},
\end{align}
times then the trace distance between any input and output is at most $\epsilon$. 
This is to be compared with e.g. Eq.~\eqref{equ:maximumdepth} of the main text, which discussed the maximum distance after which we can approximate the energy of the output of a noisy circuit.  If we take e.g. $\sigma=I/2^n$ and we take $\alpha$ and $\epsilon$ to be constant, then we are $\epsilon$-close to $\sigma$ after $\log(n)$ depth. On the other hand, Eq.~\eqref{equ:maximumdepth} asserts that we can already approximate the energies of the output with a classically simulatable state after constant depth. This difference in the scaling clarifies that our analysis differs from estimating when the output of the circuit is essentially maximally mixed.

Although one might at first think that this is just a consequence of the method we used to bound the trace distance, it is possible to show that logarithmic time is really necessary before the output of the circuit becomes close to maximally mixed~\cite{PhysRevLett.103.080501,Kastoryano_2012}. The same holds when analysing how long it takes for the output to become separable~\cite{Lami_2016,Hanson_2020}, again highlighting how our analysis is qualitatively different.

\subsection{Variational formulation of the relative entropy}\label{sec:variational}

Another fundamental ingredient to concretely assess the potential of noisy quantum computers is the variational formulation of the relative entropy.
As proved by Petz~\cite{Petz_1988}, we have that:
\begin{align}\label{equ:variationalentropy}
&D(\rho\|\sigma)=\nonumber\\&\sup\limits_{\omega>0,\tr\lb\omega\rb=1}\tr\lb \rho\log(\omega)\rb-\log\lb \tr(\exp\lb \log(\omega)+\log(\sigma)\rb\rb.
\end{align}
From this we can immediately derive:
\begin{prop}[Lower bound of output energy of noisy device]\label{prop:variationalentropy}
Let $\Phi(\rho)$ be the output of a noisy quantum circuit and $H$ a Hamiltonian. For any $\sigma>0$ and $\beta>0$ define $\mathcal{Z}_{\beta,\sigma}=\tr\lb e^{-\beta H+\log(\sigma)}\rb$. Then:
\begin{align}
\tr\lb \Phi(\rho)H\rb\geq \sup\limits_{\beta>0}\beta^{-1}\lb \log(\mathcal{Z}_{\beta,\sigma})-D(\Phi(\rho)\|\sigma)\rb.\label{equ:variationallower}
\end{align}
\end{prop}
\begin{proof}
Picking $\omega=e^{-\beta H}/tr\lb e^{-\beta H}\rb$ in Eq.~\eqref{equ:variationalentropy} we obtain:
\begin{align*}
D(\Phi(\rho)\|\sigma)\geq -\beta\tr\lb \rho H\rb-\log(\mathcal{Z}_\beta)-\log(\mathcal{Z}_{\beta,\sigma}/\mathcal{Z}_\beta).
\end{align*}
The claim then follows from a straightforward manipulation of the terms above and noting that it holds for all $\beta>0$.
\end{proof}
As usual, it is instructive to take the case $\sigma=I/2^n$ into consideration to gain some intuition on the scaling of the bound above. It can then be rewritten as 
\begin{align*}
\tr\lb\Phi(\rho) H\rb\geq \beta^{-1}\lb n-\log(\mathcal{Z_\beta})-D(\Phi(\rho)\|I/2^n)\rb.
\end{align*}
Note that in the limit $\beta\to\infty$ we obtain the inequality $\tr\lb\rho H\rb\geq E_0$, where $E_0$ is the ground state energy. Moreover, note that
\begin{align*}
\lim\limits_{\beta\to0}\beta^{-1}\lb n-\log(\mathcal{Z_\beta})\rb=\tr\lb H\frac{I}{2^n}\rb.
\end{align*}
That is, for $\beta$ small, the first term is given by the energy of the maximally mixed state. Let us now once again resort to relative entropy convergence techniques.
We know that if $\Phi$ is a depth $D$ quantum circuit affected by depolarizing noise with probability $p$ we have:
\begin{align*}
 \tr\lb\rho H\rb\geq \beta^{-1}\lb n-\log(\mathcal{Z_\beta})-(1-p)^{2D}n\rb.
\end{align*}
Now assume that $(1-p)^{2D}\ll\beta$ for some small $\beta$. Then the second term will be under control for that $\beta$ and we have the promise that the first term is close to the energy of the maximally mixed state, ensuring that we do not deviate much from it. However, this inequality also provides nontrivial bounds even for moderate values of $(1-p)^{2D}$. Note that evaluating this bound requires us to estimate the partition function in some range of parameters $\beta$. However, as remarked in the main text, there are many efficient algorithms that achieve this for $\beta$ below a critical temperature. 

It is also worth noting that for local Hamiltonians one has
\begin{align}\label{equ:taylornoise}
\beta^{-1}\lb n-\log(\mathcal{Z_\beta})\rb=\tr\lb H\frac{I}{2^n}\rb-\cO(\beta n).
\end{align}
This implies that usually taking $D_{\max}=-\log(\epsilon^{-1})/\log(1-p)$ should be enough to ensure that the energy of the output of the circuit is $\cO(\epsilon n)$ away from the energy of the maximally mixed state.
Indeed, we have:
\begin{align*}
 &\tr\lb\rho H\rb\geq \epsilon^{-1}\lb n-\log(\mathcal{Z_\epsilon})-(1-p)^{2D_{\max}}n\rb\\
& =\tr\lb H\frac{I}{2^n}\rb-\cO(\beta n)-\epsilon^{-1}\epsilon^{2}n=\tr\lb H\frac{I}{2^n}\rb-\cO(\epsilon n).
\end{align*}

Although using Prop.~\ref{prop:variationalentropy} does not immediately give us classical simulability thresholds, as it still requires us to evaluate the partition function of the model, it is very useful to assess the potential of a quantum computer for a given problem. Indeed, suppose that a classical algorithm, not necessarily simulated annealing, provides us an energy $E_c$. By evaluating the partition function for some range of $\beta$ and estimating the relative entropy with a strong data processing inequality it then immediately gives us access to concrete lower bounds on the outputs of noisy quantum devices at a certain depth. In particular, we can estimate the maximal depth before the output of the device is larger than $E_c$ and the quantum computer becomes useless. We exemplify this in Section~\ref{sec:variationalexample}.

\subsection{Proof of Lemma~\ref{thm:mirrordescent} for maximally mixed states}\label{subsec:newLemma2}

However, in the case of $\sigma=I/2^n$, it is straightforward to obtain approximate classical simulability bounds from the variational formulation of the relative entropy:
\begin{cor}\label{cor:variationalmaxmixed}
Let $\rho$ be a $n$-qubit state such that $D(\rho\|I/2^n)\leq \beta\epsilon \|H\|$ for some $\epsilon,\beta>0$ and Hamiltonian $H$. Then we have:
\begin{align*}
    \tr\lb H\rho\rb\geq \tr \lb \sigma_\beta H\rb-\epsilon \|H\|.
\end{align*}
\end{cor}
\begin{proof}
Note that the log-partition function $f(\beta)=\log(\mathcal{Z}_\beta)$ is always differentiable for finite dimensional systems. Thus, by the mean value theorem we have that:
\begin{align*}
    \frac{f(0)-f(\beta)}{\beta}=-f'(\tilde{\beta})
\end{align*}
for some $\tilde{\beta}\in(0,\beta)$. It is well-known that $-f'(\tilde{\beta})=\tr \lb \sigma_{\tilde{\beta}} H\rb$. As the energy is monotone decreasing in $\beta$, we have $\tr \lb \sigma_{\tilde{\beta}} H\rb\geq\tr \lb \sigma_{\beta} H\rb $. Thus, it follows from Prop.~\ref{prop:variationalentropy} that
\begin{align*}
    \tr\lb H\rho\rb\geq\tr \lb \sigma_{\tilde{\beta}} H\rb-\beta^{-1}D(\rho\|I/2^n)\geq \tr \lb \sigma_{\tilde{\beta}} H\rb-\epsilon \|H\|,
\end{align*}
where in the last line we used our assumed bound on $D(\rho\|I/2^n)$.
\end{proof}
We will later generalize the statement above to arbitrary stationary states in Lemma~\ref{thm:mirrordescent2}.
However, the corollary above improves the estimate of Lemma~\ref{thm:mirrordescent2} by a factor of $4$. 

\subsubsection{Choice of error parameter $\epsilon$}\label{subsubsec:epsilon}

Let us now discuss the scaling and the choice of the parameter $\epsilon$ in Lemma~\ref{thm:mirrordescent} and similar statements throughout the text. It controls the precision with which we approximate the energy, as in Eq.~\eqref{equ:goodenergy}. Throughout this article the reader should think of $\epsilon$ in the range $10^{-1}-10^{-2}$. This is justified for several reasons. In the setting of the variational quantum eigensolver~\cite{moll2017quantum, kandala2017hardware, 
wang2019accelerated}, note that most algorithms have an inverse polynomial scaling in $\epsilon$ in the sample complexity for each iteration~\cite{harrow2019lowdepth}. In most of the cases, the dependency is quadratic, while other proposals get a better dependence by proposing circuits of larger depth or more complex measurements~\cite{wang2019accelerated,harrow2019lowdepth}. Thus, for most of these proposals, a precision beyond $\epsilon\sim10^{-2}$ would require a significant number of measurements for each iteration or larger circuit depths. In the adiabatic setting, the more general theorems~\cite{Jansen_2007} have a scaling of the runtime that has an $\epsilon^{-1}$ dependency on the error $\epsilon$ as defined in Eq.~\eqref{equ:goodenergy}.
Thus, obtaining high precision solutions may take very long evolution times, although some adiabatic theorems have a logarithmic dependency on the error under stronger assumptions~\cite{RevModPhys.90.015002,Ge_2016}. 
More importantly, through the variational bound of Eq.~\ref{equ:variationalentropy} we numerically observe that picking $\epsilon$ of order $10^{-2}-10^{-1}$ is actually sufficient to ensure that the energy of $\tilde{\sigma}$ is strictly smaller than that of $\Phi(\rho)$.
This behaviour is justified by eq.~\eqref{equ:taylornoise}.

\subsection{Worked through example of the variational bound}\label{sec:variationalexample}
Let us now exemplify how to apply the lower bound given in Lemma~\ref{prop:variationalentropy}. We do this with MAXCUT instances on $3$-regular graphs. We will compare the performance of the device with the SDP relaxation and Gibbs sampling with inverse temperatures in the polynomial time range.
See Fig.~\ref{fig:lowervariational} for an example.
\begin{figure}[h!]
\centering
\includegraphics[width=1\columnwidth]{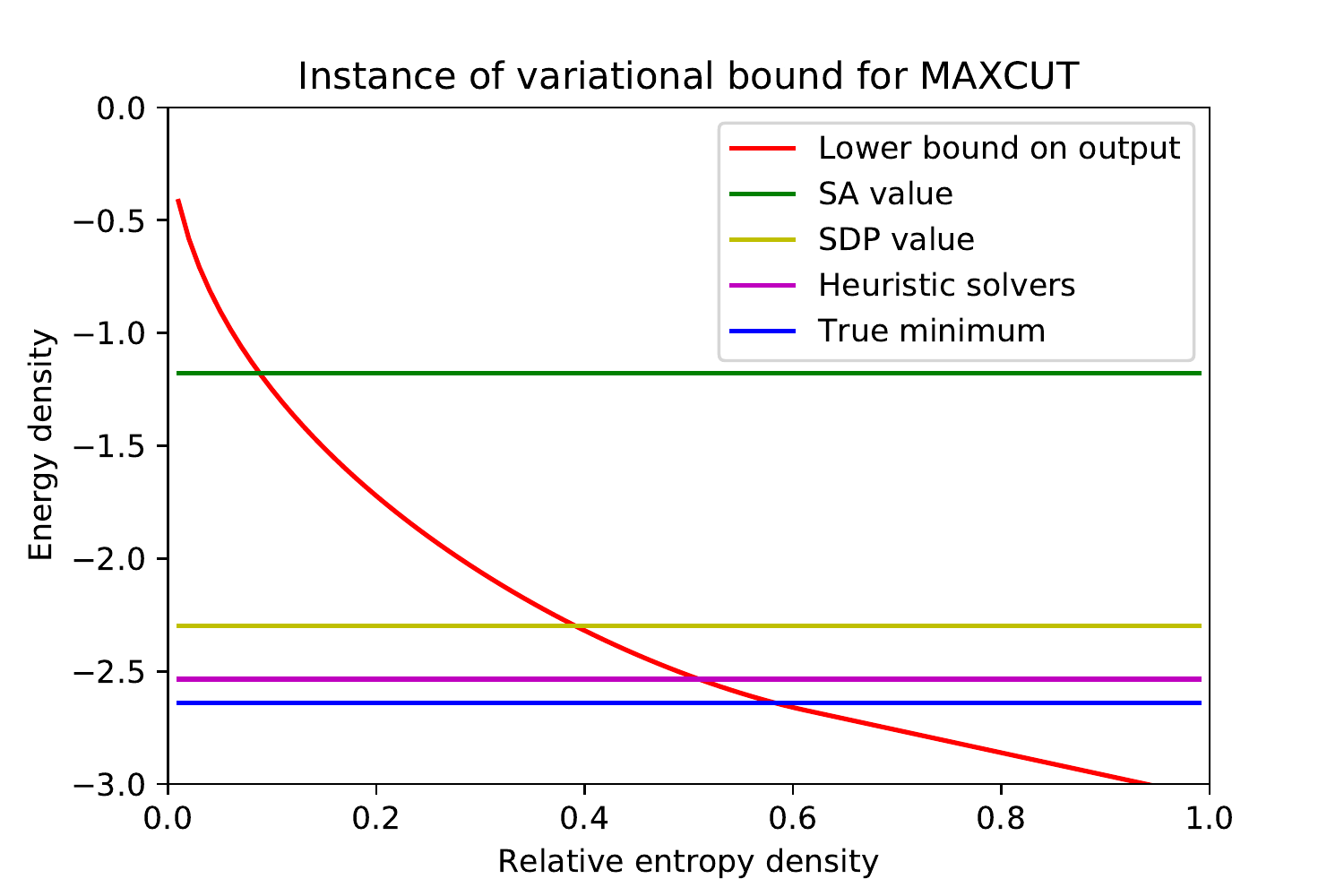}
\caption{Performance of Lemma~\ref{prop:variationalentropy} in one instance of MAXCUT of a $3$-regular graph with $20$ nodes. The lower bound on curve gives the maximum energy density predicted by Eq.~\eqref{equ:variationallower} when evaluated on $\beta\in(0,\beta_c)$ and different values of the relative entropy density. The SA value line gives the energy achieved by simulated annealing in the polynomial time range, while the SDP line gives the expected value of the SDP relaxation randomized rounding. Furthermore, the heuristic solver line gives the bound assuming that heuristic solvers achieve a 0.95 approximation ratio. From this graph we see that as long as the depth is such that the relative entropy density of the output states is at most $\simeq 0.4$, then SDP methods will yield better approximations than the noisy circuit. On the other hand, SA will outperform the noisy quantum computer at densities $\simeq 0.1$. }
\label{fig:lowervariational}
\end{figure}
As soon as one of the horizontal lines corresponding to a method crosses the lower bound line and the lower bound is then above it, it means that for relative entropy densities smaller than that threshold the noisy quantum computer is guaranteed to be outperformed by said method. 
From Fig.~\ref{fig:lowervariational}, we infer that the SDP threshold is around $0.4n$.  Thus, we get the condition $D(\Phi(\rho)\|\sigma)\leq 0.4n$ for $\sigma=I/2^n$ for SDPs to outperform the noisy quantum computer. Making the approximation $\log(0.4)\simeq -0.92$ and resorting to the usual bound $D(\Phi(\rho)\|\sigma)\leq (1-p_1)^{2Df_1}(1-p_2)^{2Df_2}$, where $f_1$ is the fraction of $1-$qubit layers with depolarizing noise $p_1$ and $f_2$, $p_2$ defined analogously for $2$ qubits, we conclude that the maximum depth satisfies
\begin{align*}
    D\leq 0.46(f_1p_1+f_2p_2)^{-1}.
\end{align*}
Note that this estimate is superior to the naive $1/p$ estimate of the maximal depth one can achieve.
Assuming the fractions of layers to be the same and the noise parameters as in Sec.~\ref{sec:numberofcycles}, we get a total depth of roughly $D\simeq 70$ of when the noisy quantum computer is guaranteed to be outperformed by the SDP relaxation. We note that we observe a threshold relative entropy density for varying graph sizes and random instances of up to $20$ qubits of around $0.35-0.4$. Thus, we will extrapolate the findings to larger system sizes. 
Assuming the slightly more conservative relative entropy density bound threshold of $0.3$ scales to larger system sizes and that each round of QAOA requires a depth scaling like $n$, we see that we expect the SDP to outperform three rounds of QAOA already at systems of size $30$, orders of magnitude below what current solvers can handle.

\section{Correlation length and required circuit depth}\label{sec:corlength}

In Eq.~\eqref{equ:maximumdepth} we gave a bound on the maximal circuit depth before a noisy quantum device is outperformed by classical methods that only depended on the geometry of the lattice and noise rate. We will now show that this also implies that noisy quantum devices performing VQE are unlikely to find the ground states with long correlation length. First, recall that given a graph $G=(V,E)$, a state is said to have correlation length $\xi$ if for all observables $O_x$ and $O_y$ supported on sites $x,y\in V$ we have that
\begin{align}\label{equ:corrlengh}
&|\tr\lb\ketbra{\psi}{\psi} O_x\otimes O_y\rb-\tr\lb\ketbra{\psi}{\psi} O_x\rb\tr\lb\ketbra{\psi}{\psi}O_y\rb|\\&\leq 2\|O_x\|\|O_y\|e^{-\frac{d(x,y)}{\xi}}.
\end{align}
Let us assume a negation of Eq~\eqref{equ:corrlengh}, i.e. the existence of observables on $O_x,O_y$ at $x,y$ at distance $\xi$ s.t. 
\begin{align}\label{equ:corrlengh}
&|\tr\lb\ketbra{\psi}{\psi} O_x\otimes O_y\rb-\tr\lb\ketbra{\psi}{\psi} O_x\rb\tr\lb\ketbra{\psi}{\psi}O_y\rb|\\&\geq \|O_x\|\|O_y\|C,
\end{align}
where $C\gg 0$ is some constant. That is, there are significant correlations in the state at distance $d(x,y)$. Suppose further for simplicity that the graph $G$  also reflects the connectivity of the noisy quantum device and let $U=U_1U_2\cdots U_D$ be a circuit of depth $D$.
Let us recall the well-known result that, starting from a product state, a depth proportional to the correlation length is required to prepare a state:
\begin{lem}
Let $\ket{\psi}$ be the ground state of $\kappa$-local Hamiltonian $H$ on a $d$-dimensional lattice such that there exist observables $O_x,O_y$ with $\|O_x\|,O_y\|=1$ supported at sites $x,y$ at distance $\xi$ fulfilling 
\begin{align}\label{equ:corrlengh2}
&\tr\lb\ketbra{\psi}{\psi} O_x\otimes O_y\rb-\tr\lb\ketbra{\psi}{\psi} O_x\rb\tr\lb\ketbra{\psi}{\psi}O_y\rb\nonumber\\&\geq C.
\end{align}
Then the output $\Phi\lb \ketbra{0}{0}^{\otimes n}\rb$ of any depth $\frac{\xi}{2}$ circuit $\Phi$ with the same locality as the lattice satisfies:
\begin{align*}
\|\Phi(\ketbra{0}{0}^{\otimes n})-\ketbra{\psi}{\psi}\|_{tr}\geq \frac{C}{4}
\end{align*}
\end{lem}
\begin{proof}
We will prove the claim by contradiction. Assume that 
\begin{align}\label{equ:statesareclose}
\|\Phi(\ketbra{0}{0}^{\otimes n})-\ketbra{\psi}{\psi}\|_{tr}< \frac{C}{4}
\end{align}
By the variational formulation of the trace distance, it suffices to find an observable $O$ with $\|O\|\leq 1$ and such that
\begin{align*}
\tr\lb O\lb \Phi(\ketbra{0}{0}^{\otimes n})-\ketbra{\psi}{\psi}\rb\rb\geq\frac{C}{4}
\end{align*}
for a contradiction.
Consider
\begin{align*}
O=O_x\otimes O_y
\end{align*}
Note that as the circuit has the same locality as the lattice and sites $x$ and $y$ are a distance $\xi$ apart, we have that $\Phi^*(O_x)$ and $\Phi^*(O_y)$ have disjoint support. Here $\Phi^*$ corresponds to the circuit in the Heisenberg picture. This is because the circuit has depth at most $\tfrac{\xi}{2}$. 
Let $a_x=\tr\lb O_x\Phi(\ketbra{0}{0}^{\otimes n})\rb$ and $a_y=\tr\lb O_y\Phi(\ketbra{0}{0}^{\otimes n})\rb$.
Thus, 
\begin{align*}
\tr(O\Phi(\ketbra{0}{0}^{\otimes n}))= a_xa_y
\end{align*}
Let $b_{x,y}=\tr\lb O_x\otimes O_y\ketbra{\psi}{\psi}\rb$, $b_x=\tr\lb O_x\ketbra{\psi}{\psi}\rb$ and $b_y=\tr\lb O_y\ketbra{\psi}{\psi}\rb$. By Eq.~\eqref{equ:corrlengh2}, we have that:
\begin{align*}
b_{x,y}\geq C+b_xb_y.
\end{align*}
Moreover, as we assumed Eq.~\eqref{equ:statesareclose}, we have that
\begin{align*}
|b_x-a_x|<\frac{C}{4},\quad |b_y-a_y|<\frac{C}{4}.
\end{align*}
Thus:
\begin{align*}
&\tr\lb O\lb \Phi(\ketbra{0}{0}^{\otimes n})-\ketbra{\psi}{\psi}\rb\rb\geq C+b_xb_y-a_xa_y\\
&\geq C-(|b_x|+|b_y|)\frac{C}{4}\geq \frac{C}{2}.
\end{align*}
This yields the claim.
\end{proof}

Thus, we see that for depths $D\leq \xi/2$, the quantum circuit is guaranteed to be a constant distance away from the ground state and conclude that a depth which scales at least linearly with the correlation length of the state is necessary to faithfully reproduce the ground state. 
This, however, does not exclude the possibility that a quantum device can outperform classical optimization methods at smaller depths if there are excited states with significantly smaller correlation lengths and small energies. But then, if we are interested in properties of the ground state, it is arguable that the quantum computer is outperforming classical methods for the wrong reason. This is because the output of the circuit will necessarily be far away from the ground state. And, at larger depths, unless the noise rate is sufficiently small, the device is outperformed by classical methods in optimizing the energy. This is captured in the next proposition:
\begin{prop}\label{prop:noisycorrel}
Let $\ket{\psi}$ be the ground state of a $\kappa$ local Hamiltonian $H$ on a $d$-dimensional lattice  such that there exist observables $O_x,O_y$ at distance $\xi$ fulfilling 
\begin{align}\label{equ:corrlengh}
&|\tr\lb\ketbra{\psi}{\psi} O_x\otimes O_y\rb-\tr\lb\ketbra{\psi}{\psi} O_x\rb\tr\lb\ketbra{\psi}{\psi}O_y\rb|\nonumber\\&\geq \|O_x\|\|O_y\|C.
\end{align}
Suppose that $\Phi(\ketbra{0}{0}^{\otimes n})$ is the output of a noisy quantum circuit with depth $D$ respecting the locality of  the lattice. Furthermore, suppose that the circuit undergoes local depolarizing noise with rate $p$ that satisfies:
\begin{align*}
p\geq\frac{2\log(20e d^\kappa\epsilon^{-1})}{\xi}.
\end{align*}
Then $\Phi(\ketbra{0}{0}^{\otimes n})$ either satisfies:
\begin{align}\label{equ:disstancelarge}
\|\Phi(\ketbra{0}{0}^{\otimes n})-\ketbra{\psi}{\psi}\|_{tr}\geq \frac{C}{4}
\end{align}
or there is a quasi-polynomial time classical algorithm that outputs a state $\tilde{\sigma}$ that satisfies:
\begin{align}
    \tr\lb H\lb \tilde{\sigma}-\Phi(\rho)\rb\rb\leq \|H\|\epsilon.
\end{align}
\end{prop}
\begin{proof}
If the depth of the circuit $\Phi$ is less than $\frac{\xi}{2}$, then Eq.~\eqref{equ:disstancelarge} holds. On the other hand, for depths $D\geq \frac{\xi}{2}$, we have a quasi-polynomial algorithm by Proposition~\ref{prop:depth_classical} and the bound on the noise rate.
\end{proof}

Unfortunately, it is a hard task to show that for a given Hamiltonian there are no low-energy states that have a small correlation length and, thus, we cannot exclude the possibility of there being minima of the energy at low depths. However, for some models, one can show that a quantum circuit of at least logarithmic depth is required to beat polynomial time classical algorithms.
For instance, it has been shown that the quantum approximate optimization algorithm requires at least logarithmic depth to outperform classical polynomial time approximation algorithms to optimize the energy of certain classical Ising Hamiltonians~\cite{2005.08747,1910.08980}. This implies that even noise rates that decay with the number of qubits are not enough to ensure that QAOA outperforms classical methods:
\begin{prop}[Inverse logarithmic noise QAOA for classical Ising models never outperforms polynomial time classical for certain instances]
Suppose a quantum device with $n$ qubits suffers from $1$-local depolarizing noise after each layer of QAOA is applied with probability
\begin{align*}
p\geq \frac{\log\lb \epsilon^{-1}\rb\log(\Delta-1)}{2\log(n)}
\end{align*}
for some $\Delta\geq3$ and $\epsilon>0$. Then there exist $\Delta$-regular graphs $G=(V,E)$ on $n$ vertices such that for the Hamiltonian
\begin{align*}
H=-\sum\limits_{i\sim j}Z_iZ_j
\end{align*}
performing QAOA on the noisy device never outperforms polynomial time randomized classical algorithms by more than $\epsilon \Delta n$.
\end{prop}
\begin{proof}
In~\cite{1910.08980,2005.08747} the authors identify instances such that at least
\begin{align}\label{equ:roundsQAOA}
    D_{\min}=\log(n)/\log(\Delta-1)
\end{align}
applications of the QAOA unitary operator are required to achieve an approximation ratio that is better than $1/2$. Thus, for depths smaller than $D_{\min}$, QAOA is outperformed by the Goemanson-Williamson algorithm~\cite{Goemans_1995}, as it achieves an approximation ratio greater than $0.87$. 
As we assumed local depolarizing noise, if $\Phi$ is a noisy depth $D$ quantum channel we have:
\begin{align*}
D\lb\Phi(\rho)\|\frac{I}{2^n}\rb\leq  (1-p)^{2D}n.
\end{align*}
Thus, to obtain an approximation up to $\epsilon n\Delta$, we need to sample from the inverse temperature
\begin{align*}
\beta=\frac{(1-p)^{2D}}{\Delta\epsilon}.
\end{align*}
The results of~\cite{2007.08200} assert that we may sample efficiently from this inverse temperature as long as
\begin{align*}
\beta\leq\frac{1}{\Delta}.
\end{align*}
Thus, if $p$ is such that for $D_{\min}$ we have 
\begin{align*}
\lb\frac{(1-p)^{2D_{\min}}}{\epsilon}\rb\leq 1,
\end{align*}
then also for depths larger than $D_{\min}$ the noisy quantum device will not substantially outperform classical methods.
Solving for $p$ we then obtain the constraint:
\begin{align*}
\log(1-p)\leq \frac{\log(\Delta-1)}{2\log(n)}\log\lb\epsilon^{-1}\rb 
\end{align*}
Using $\log(1-p)\geq-p$ yields the claim.
\end{proof}

Thus, we see that without error correction, for some instances QAOA is outperformed by classical methods even with noise rates decaying logarithmically with system size. To the best of our knowledge, this is the first result that asserts that a noisy quantum algorithm will be outperformed by polynomial time-classical algorithms even with local noise rates converging to $0$. However, as pointed out in~\cite{2005.08747}, the logarithmic dependency in system size on the depth of the circuit is not too restrictive in practice and this bound would admittedly only be relevant for very large system sizes. However, as remarked in the main text, if the depth of each QAOA layer scales with system size due to the limited connectivity of the device, then the bound above can actually be used to obtain rigorous results on how noisy QAOA will never substantially outperform polynomial time algorithms for some instances. Indeed, evaluating Eq.~\eqref{equ:roundsQAOA} for $n=10^3$ qubits and $\Delta=3$ gives that we need roughly $10$ rounds of QAOA to outperform SDPs. Note that for the estimates in Sec.~\ref{sec:numberofcycles} we only assumed three rounds of QAOA.
Thus, performing $10$ rounds of QAOA with each QAOA unitary requiring circuit depth scaling mildly with $n$ would require extremely low noise rates in order for results like Prop.~\ref{prop:ising} not to guarantee classical superiority. 

Although we cannot prove the lack of advantage at any depth for noisy quantum algorithms beyond the example above, a striking phenomenon was observed in~\cite{BravoPrieto2020scalingof}. There, the authors show numerical evidence that for critical systems, which have long correlation lengths, noiseless VQE and QAOA algorithms do not make significant progress in lowering the energy before depths proportional to the system size. This suggests that an even stronger variation of Prop.~\eqref{prop:noisycorrel} may be true for such systems. That is, the noisy quantum computer will not output lower energy states for any depth unless the noise rate scales at least inverse linearly with system size.

On the other hand, for gapped systems, which have exponentially decaying correlations~\cite{Hastings2006}, the authors of~\cite{BravoPrieto2020scalingof} show that noiseless VQE algorithms already start exponentially converging to the ground state for small depths. 

Thus, if the behaviour observed in~\cite{BravoPrieto2020scalingof} for the preparation of states with long correlation length is indeed universal for VQE algorithms, it may be difficult to demonstrate any advantage in optimizing the energy for states with long correlation lengths. And, if there indeed is one, then it may also happen because the VQE algorithm is stuck in a local minimum, as discussed in Prop.~\ref{prop:noisycorrel}. 

We conclude that our results indicate that it is more realistic to expect noisy quantum computers to prepare states with a small correlation length. However, these are exactly the class of states classical simulation techniques such as tensor networks~\cite{Or_s_2014} excel at. Moreover, as argued in~\cite{kim2017robust,kim2017noiseresilient,1909.04786}, for such states it is more noise robust to prepare them in a sequential fashion instead of through a quantum circuit. Thus, when it comes to optimizing the ground state of local Hamiltonians on noisy devices, it seems advisable to focus on Hamiltonians that are expected to be gapped, but still out of reach for current classical methods. Furthermore, it may be advisable to focus on sequential preparation instead of a circuit ansatz.

\section{General Circuits}\label{sec:GeneralCircuits}
In this section we show how to generalize Lemma~\ref{thm:mirrordescent} beyond the maximally mixed case.
Furthermore, we develop some tools that lead to the qubit circuit equivalent of Lemma~\ref{thm:concentration} in the main manuscript 
and that is a preliminary step to derive a variation of entropic convergence for the quantum annealers beyond the case of a maximally mixed fixed point. This is the content of Lemma~\ref{lem:dataprocessedtriangle}.

\subsection{Mirror descent}\label{sec:mirrordescent}

Now that we discussed some basic properties of the relative entropy, let us discuss Lemma~\ref{thm:mirrordescent} in more detail and prove an analogue beyond the maximally mixed state.
As mentioned in the main text, the goal of our framework is to find a state that approximates the energy of the  noisy quantum device. To that end, given a quantum state $\rho\in\mathcal{D}_{2^n}$, an $n$-qubit Hamiltonian $H$ and an error parameter $\epsilon>0$, define the sets:
\begin{align*}
\cC(\rho,H,\epsilon)=\{\sigma\in \mathcal{D}_{2^n}:\tr\left(H\left(\sigma-\rho\right)\right)\leq \epsilon \|H\|\}.
\end{align*}
That is, $\sigma\in\cC(\rho,H,\epsilon)$ if its energy is at most $\epsilon \|H\|$ larger than that of $\rho$. We will then say that $\sigma$ performs approximately as well as $\Phi(\rho)$ to optimize $H$, where $\Phi$ is as usual a noisy quantum device with input $\rho$, if $\sigma\in\cC(\Phi(\rho),H,\epsilon)$.
The goal of our framework will be to find a Gibbs state $\tilde{\sigma}$ that is in $\cC(\Phi(\rho),H,\epsilon)$ for a given device and noise levels.
We will now prove the convergence of the update rule of mirror descent for completeness:
\begin{lem} \label{lem:update-rule}
Let $\rho$ be a quantum state.
Fix a Hamiltonian $H_0$ and let $\sigma_0= e^{-H_0}/\mathcal{Z}_0$. 
Suppose that for some other Hamiltonian $H$ we have:
\begin{align*}
\tr\lb H\lb \sigma_0-\rho\rb\rb\geq \|H\|\epsilon.
\end{align*}
Set $H_{1}=H_0+\frac{\epsilon}{2\|H\|}H$. Then, the Gibbs states $\sigma_1= e^{-H_1}/\mathcal{Z}_1$ obeys
\begin{equation*}
D(\rho \| \sigma_{1}) - D(\rho \|\sigma_0) \leq- \frac{\epsilon^2}{4}
\end{equation*}
\end{lem}

\begin{proof}
We have:
\begin{align*}
&D(\rho \| \sigma_{1}) - D (\rho \| \sigma_0) 
=\\& \mathrm{tr} \left( \rho (H_{1}-H_0 \right)
+ \log \left( \frac{\mathrm{tr}(\exp (-H_{1}))}{\mathrm{tr}(\exp (- H_0))} \right)
\end{align*}
By construction, $H_1-H_{0}=\frac{\epsilon}{2\|H\|}H$ and the first term equals $\frac{\epsilon}{2\|H\|} \mathrm{tr}(H \rho)$. 
The logarithmic ratio can be bounded by $\frac{\epsilon}{2\|H\|}\mathrm{tr}(H \sigma_1)$ and we refer to \cite{fern2019faster}[Proof of Lemma~3.1] for a derivation based on the Peierls-Bogoliubov inequality. 
Adding and subtracting $\frac{\epsilon}{2\|H\|}\mathrm{tr}(H \sigma_0)$, we conclude
\begin{align*}
&D(\rho \| \sigma_{1}) - D (\rho \| \sigma_0) 
 \\
\leq & \frac{\epsilon}{2}
\left( \mathrm{tr}\lb\frac{H}{\|H\|}(\sigma_{1}-\sigma_0)\rb + \mathrm{tr}\lb\frac{H}{\|H\|}(\rho-\sigma_0)\rb \right) \\
\leq & \frac{\epsilon}{2} \left( \tfrac{1}{2}\| \sigma_{1}-\sigma_0 \|_{\mathrm{tr}}+\mathrm{tr}\lb \frac{H}{\|H\|}(\rho-\sigma_0\rb\right),
\end{align*}
Here we used Hellstrom's bound to get
\begin{align*}
\tfrac{1}{2}\| \sigma_{1}-\sigma_0 \|_{\mathrm{tr}}\geq  \mathrm{tr}\lb\frac{H}{\|H\|}(\sigma_{1}-\sigma_0)\rb.
\end{align*}
In terms of 
The claim now follows from noticing that $\sigma_0$ and $\sigma_1$ must be close in trace distance. \cite[Lem.~16]{Brandao2017a} implies 
$\tfrac{1}{2}\| \sigma_{1} - \sigma_0 \|_{tr} \leq \exp (\frac{\epsilon}{4}) -1 \leq  \frac{\epsilon}{2}$ .
\end{proof}
That is, as long as we have not converged to a state in $\tilde{\sigma}\in\cC(\rho,H,\epsilon)$, updating the Hamiltonian of the Gibbs state gives rise to a state that is closer to the target state in relative entropy.
Let us now see how this can be used to prove Lemma~\ref{thm:mirrordescent}, which we now restate here for the reader's convenience. Before stating the Lemma, let us remark that the statement below is worse by a factor of $1/4$  for the required relative entropy compared to that of Lemma~\ref{thm:mirrordescent}. That is, here we required the Gibbs state with Hamiltonian $\frac{4\lambda}{\|H\|\epsilon}H$ instead of $\frac{\lambda}{\|H\|\epsilon}H$. This is because in the case of $\sigma$ being maximally mixed state it is straightforward to obtain this slightly improved statement from the formulation of the relative entropy, as we show later in Cor.~\ref{cor:variationalmaxmixed}.
\begin{lem}\label{thm:mirrordescent2}
Let $\Phi(\rho)$ be the output of a noisy quantum device and for some quantum state $\sigma>0$ let $H_\infty=\log(\sigma)$ and 
$D(\Phi(\rho)||\sigma)=\tr\lb \Phi(\rho)\lb \log(\Phi(\rho))-\log(\sigma)\rb\rb$ be their relative entropy.
Then for any Hamiltonian $H$ there is a $\lambda\in [0,D(\Phi(\rho)||\sigma)]$ such that the state $\tilde{\sigma}$ defined as
\begin{align}\label{equ:thermalstateclassical}
    \tilde{\sigma}=\text{exp}\lb -H_\infty-\frac{4\lambda}{\|H\|\epsilon}H\rb/\mathcal{Z}
\end{align}
satisfies:
\begin{align}\label{equ:goodenergy}
    \tr\lb H\lb \tilde{\sigma}-\Phi(\rho)\rb\rb\leq \|H\|\epsilon.
\end{align}
That is, $\tilde{\sigma}\in\cC(\rho,H,\epsilon)$.
\end{lem}
\begin{proof}
We will prove this by contradiction. Suppose that for no $\lambda\in[0,D(\Phi(\rho)||\sigma)]$ we have that $\tilde{\sigma}$ satisfies Eq.~\eqref{equ:goodenergy}. Define the sequence of states
\begin{align*}
\sigma_t=\text{exp}\lb -H_\infty-\frac{t\epsilon}{4\|H\|}H\rb/\mathcal{Z}_t
\end{align*}
for $t\geq 0$. Note that by our assumption, we have for all $t\in [0,4\epsilon^{-2}D(\Phi(\rho)||\sigma)]$ that:
\begin{align*}
   \tr\lb H\lb \sigma_t-\Phi(\rho)\rb\rb> \|H\|\epsilon.
\end{align*}
Thus, it follows from~\ref{lem:update-rule} that:
\begin{align*}
D(\Phi(\rho) \| \sigma_{t+1}) - D(\Phi(\rho) \|\sigma_{t}) \leq- \frac{\epsilon^2}{4}.
\end{align*}
We conclude that 
\begin{align}\label{equ:approxrelativeentropy}
D(\Phi(\rho)|| \sigma_t)\leq- t\frac{\epsilon^2}{4} +D(\Phi(\rho)||\sigma_0).
\end{align}
If we pick $t=\lceil 4\epsilon^{-2}D(\Phi(\rho)||\sigma)\rceil$, we get that
\begin{align*}
D(\Phi(\rho)|| \sigma_t)<0,
\end{align*}
which contradicts the positivity of the relative entropy.
\end{proof}
Note that the Lemma above only ensures that the state $\tilde{\sigma}$ reproduces approximately the same energy as the output of the noisy circuit. It does not promise us that the state $\tilde{\sigma}$ approximates the output $\Phi(\rho)$  for other observables.
This is because as soon as $\tilde{\sigma}\in \cC(\Phi(\rho),H,\epsilon)$, we do not have any information on how to update the Hamiltonian of $\tilde{\sigma}$ solely based on the energy and further expectation values are required to ensure  convergence.

However, a close inspection of the proof above shows that if the value  of $t$ we need to observe comparable energies to $\Phi(\rho)$ is high, then we are ensured that $\tilde{\sigma}$ is a good approximation to $\Phi(\rho)$. Indeed, we see from Eq.~\eqref{equ:approxrelativeentropy} that:
\begin{align*}
&\frac{1}{2}\|\Phi(\rho)-\sigma_t\|_{\tr}^2\leq D(\Phi(\rho)|| \sigma_t)\\&\leq- t\frac{\epsilon^2}{4} +D(\Phi(\rho)||\sigma_0).
\end{align*}
Thus, if the $\lambda$ required to observe similar energies to $\Phi(\rho)$ is of order $D(\Phi(\rho)||\sigma_0)$, then we are also guaranteed a good global approximation.

\subsection{Entropic convergence in quantum circuits}\label{sec:entrodiscrete}
The goal of this subsection is to obtain the entropic convergence required for our framework.
Let us start by bounding the convergence of the entropy in discrete time, i.e. in quantum circuits. The main technical result we need for this is the data-processed inequality of~\cite[Theorem III.1]{Christandl_2017}. It asserts that for a quantum channel $P$ and states $\rho,\sigma$ and $\sigma'$ we have:
\begin{align*}
D(P(\rho)\|\sigma)\leq D(\rho\|\sigma')+D_{\infty}(P(\sigma')\|\sigma).
\end{align*}
We then have:
\begin{lem}\label{lem:dataprocessedtriangle}
Let $T:\M_{2^n}\to\M_{2^n}$ be a quantum channel with fixed point $\sigma$ that satisfies a strong data-processing inequality with constant $\alpha>0$. That is,
\begin{align*}
D(T(\rho)||\sigma)\leq (1-\alpha)D(\rho||\sigma)
\end{align*}
for all states $\rho$. Then for any other quantum channels $\Phi_1,\ldots,\Phi_m:\M_{2^n}\to\M_{2^n}$ we have:
\begin{align}\label{equ:dataproccphi}
&D(\bigcirc_{t=1}^m (\
\Phi_t\circ T)(\rho)||\sigma)\leq (1-\alpha)^m D(\rho||\sigma)+\nonumber\\&\sum_{t=0}^{m-1} (1-\alpha)^{m-t} D_{\infty}( \Phi_{t}(\sigma)||\sigma),
\end{align}
where 
\begin{align*}
D_{\infty}(\rho||\sigma)=\log(\|\sigma^{-\frac{1}{2}}\rho\sigma^{-\frac{1}{2}}\|).
\end{align*}
\end{lem}
\begin{proof}
We will prove this by mathematical induction. For $m=1$, this follows from the  data-processed triangle inequality of~\cite[Theorem III.1]{Christandl_2017}. Following their notation with $P= \
\Phi$ and $\sigma'=\sigma$ we have:
\begin{align}\label{equ:resultn=1}
D( \
\Phi_1\circ T)(\rho)||\sigma)\leq D(T(\rho)||\sigma)+ D_{\infty}(\Phi_1(\sigma)||\sigma).
\end{align}
Let us now assume the  claim to be true for $m=k$ and let us show it for $m=k+1$. We have:
\begin{align}\label{equ:inductionhypok}
&D((\bigcirc_{t=1}^{k+1} (\
\Phi_t\circ T)(\rho)||\sigma)\leq (1-\alpha)^k D( \
\Phi_{k+1}\circ T)(\rho)||\sigma)+\nonumber\\&\sum_{t=0}^{k-1} (1-\alpha)^{m-t} D_{\infty}( \Phi_{t}(\sigma)||\sigma)
\end{align}
by our induction hypothesis. 
Applying Eq.~\eqref{equ:resultn=1} to the first term in Eq.~\eqref{equ:inductionhypok} and applying the strong data-processing inequality we obtain the claim.
\end{proof}
The lemma above yields a good estimate on the decay of the relative entropy in two (not mutually exclusive) scenarios. First, when the state $\sigma=e^{-\gamma H_1}/\tr(e^{-\gamma H_1})$ is a high temperature Gibbs state. Indeed, in this case we have, as mentioned before, that $D_{\infty}(\Phi(\sigma)||\sigma)\leq 2\gamma \|H_1\|$. Thus, the second term in~\eqref{equ:dataproccphi} will be small irrespective of the unitaries being implemented if $\gamma\ll \|H_1\|^{-1}$.
In particular, in the doubly stochastic case, i.e. $\sigma=\frac{I}{2^n}$, we have $D_{\infty}(\Phi(\sigma)||\sigma)=0$. 
The second scenario is the one in which the quantum circuit approximately preserves the fixed point at late times of the evolution, i.e. $\Phi_t(\sigma)\simeq \sigma$ for all $t\geq m_0$. In this case, the $(1-\alpha)^{m-t}$ prefactor makes sure that the contribution of $ D_{\infty}( \Phi_{t}(\sigma)||\sigma)$ is small for $t\leq m_0$ and $D_{\infty}( \Phi_{t}(\sigma)||\sigma)$ is small by our assumption that $\Phi_t(\sigma)\simeq \sigma$ for $t>m_0$. For instance, this could be the case for a circuit consisting overwhelmingly of diagonal gates at the end of the computation under decoherence, that is, $\sigma$ a diagonal state. A concrete example of such a circuit would be the quantum Fourier transform.

\section{Quantum Annealers}\label{sec:SMquann}

\subsection{Entropic convergence in quantum annealers}
The statement of Lemma~\ref{lem:dataprocessedtriangle} allows us to analyse the entropic convergence of circuits in discrete time. Let us now generalize these results to quantum annealers under uniform noise.
The first step will be to generalize Lemma~\ref{lem:dataprocessedtriangle} for a fixed Lindbladian and a fixed Hamiltonian:
\begin{lem}\label{lem:fixedhamiltonian}
Let $\cL:\M_{2^n}\to\M_{2^n}$ be a Lindbladian with fixed point $\sigma$ satisfying a MLSI with constant $\alpha>0$. Moreover, let $\cH:\M_{2^n}\to\M_{2^n}$ be given by $\cH(X)=i[H,X]$ for some Hamiltonian $H$. Then for all states $\rho$ and times $t>0$:
\begin{align*}
&D(e^{t(\cL+\cH)}(\rho)||\sigma)\leq e^{-\alpha t}D(\rho||\sigma)\\&+\alpha^{-1}(1-e^{-\alpha t})\|\sigma^{-\frac{1}{2}}[H,\sigma]\sigma^{-\frac{1}{2}}\|.
\end{align*}
\end{lem}
\begin{proof}
Denote by $T_m=e^{\frac{t}{m}\cL}$ and $\Phi_m= e^{  \frac{t}{m}\cH}$ for some $m\in \N$. By the Lie-
Trotter formula we have:
\begin{align*}
\lim_{m\to\infty} (\Phi_m T_m)^m=e^{t(\cL+\cH)}.
\end{align*}
By Lemma~\ref{lem:dataprocessedtriangle} we have for all $m$:
\begin{align*}
&D((\Phi_m\circ T_m)^m(\rho)||\sigma)\leq e^{-t\alpha} D(\rho||\sigma)+\\& \sum\limits_{t=1}^m e^{-\frac{\alpha  t}{m}}D_{\infty}(\Phi_m(\sigma)||\sigma),
\end{align*}
as  each $T_m$ contracts the relative entropy by $e^{-\frac{\alpha t}{m}}$. 
First, we can sum the geometric series and see that
\begin{align*}
 \sum\limits_{t=1}^m e^{-\frac{\alpha  t}{m}}=\frac{(1-(e^{-t\alpha}))}{1-e^{-\alpha  t/m}}
\end{align*}

Let us now bound the limit:
\begin{align*}
\lim_{m\to \infty}\frac{(1-(e^{-t\alpha}))}{1-e^{-\alpha  t/m}} D_{\infty}(\Phi_m(\sigma)||\sigma).
\end{align*}
By a Taylor expansion:
\begin{align*}
\frac{(1-(e^{-t\alpha/2}))}{1-e^{-\alpha  t/m}}=\frac{1-(e^{-t\alpha}))}{\alpha t +\cO(m^{-1})}m.
\end{align*}
Thus, let us bound the limit 
\begin{align*}
\lim\limits_{m\to\infty}mD_{\infty}( \Phi_m(\sigma)||\sigma).
\end{align*}
As $\log(1+x)\leq x$, we have:
\begin{align*}
    mD_{\infty}( \Phi_m(\sigma)||\sigma)\leq m\lb \|\sigma^{-\frac{1}{2}}e^{\frac{t}{m} \cH}(\sigma)\sigma^{-\frac{1}{2}}\|-1\rb.
\end{align*}
Performing a Taylor expansion of $e^{\frac{t}{m} \cH}$ followed by a triangle inequality we see that
\begin{align*}
   \|\sigma^{-\frac{1}{2}}e^{\frac{t}{m} \cH}(\sigma)\sigma^{-\frac{1}{2}}\|\leq 1+\frac{t}{m}\|\sigma^{-\frac{1}{2}}[H,\sigma]\sigma^{-\frac{1}{2}}\|+\cO\lb \frac{t^2}{m^2}\rb.
\end{align*}
We conclude that
\begin{align*}
&m\lb \|\sigma^{-\frac{1}{2}}e^{\frac{t}{m} \cH}(\sigma)\sigma^{-\frac{1}{2}}\|-1\rb\\&
\leq t\|\sigma^{-\frac{1}{2}}[H,\sigma]\sigma^{-\frac{1}{2}}\|+\cO\lb \frac{t^2}{m}\rb,
\end{align*}
which gives:
\begin{align*}
    &\lim_{m\to \infty}\frac{(1-(e^{-t\alpha}))}{1-e^{-\alpha  t/m}} D_{\infty}( \Phi_m(\sigma)||\sigma)\leq\\&
    \frac{(1-e^{-t\alpha})}{\alpha}\|\sigma^{-\frac{1}{2}}[H,\sigma]\sigma^{-\frac{1}{2}}\|
\end{align*}
\end{proof}
This allows us to control the scaling of the relative entropy with a fixed Hamiltonian term and gives the following simple corollary:
\begin{cor}\label{cor:manytimesteps}
Let $\cS_1,\cS_2,\ldots,\cS_m$ be Lindbladians of the form $\cS_l=\cL+\cH_l$. Here $\cL$ is a fixed Lindbladian with fixed point $\sigma>0$ satisfying a MLSI with constant $\alpha>0$. Then for all states $\rho$ and times $t_1,\ldots,t_m$:
\begin{align}
&D(\bigcirc_{l=1}^m e^{t_l \cS_{t_l} }(\rho)||\sigma)\leq e^{-\alpha T}D(
\rho||\sigma)\nonumber\\&+\alpha^{-1}\sum\limits_{l=1}^me^{-\alpha T_l}(1-e^{-\alpha t_l})\|\sigma^{-\frac{1}{2}}[H_l,\sigma]\sigma^{-\frac{1}{2}}\|
\end{align}
with $T_l=\sum\limits_{j=1}^{l} t_l$.
\end{cor}
\begin{proof}
We will prove the claim by mathematical induction. For $m=1$, the claim is just Lemma~\ref{lem:fixedhamiltonian}. Assume the claim is true for all $n\leq k$. It follows from our induction hypothesis that:
\begin{align*}
&D(\bigcirc_{l=1}^{k+1} e^{t_l \cS_{t_l} }(\rho)||\sigma)\leq e^{-T_k}D(e^{t_{k+1}}\cS_{k+1}(\rho)||\sigma)\\&+\alpha^{-1}\sum\limits_{l=1}^ke^{-\alpha T_l}(1-e^{-\alpha t_l})\|\sigma^{-\frac{1}{2}}[H_l,\sigma]\sigma^{-\frac{1}{2}}\|
\end{align*}
It again follows from Lemma~\ref{lem:fixedhamiltonian} that:
\begin{align*}
&e^{-T_k}D(e^{t_{k+1}}\cS_{k+1}(\rho)||\sigma)\leq e^{-T_{k+1}}D(\rho||\sigma)\\&+\alpha^{-1}\sum\limits_{l=1}^{k+1}e^{-\alpha T_l}(1-e^{-\alpha t_l})\|\sigma^{-\frac{1}{2}}[H_l,\sigma]\sigma^{-\frac{1}{2}}\|,
\end{align*}
which yields the claim.
\end{proof}
The next step is to consider the evolution of a Lindbladian with a time-dependent Hamiltonian but uniform dissipative part and bound its relative entropy. That is, let $\cL:\M_{2^n}\to\M_{2^n}$ be a Lindbladian with fixed point $\sigma$. Moreover, let $\cH_t:\M_{2^n}\to\M_{2^n}$ be given by $\cH_t(X)=i[H_t,X]$ for some time-dependent Hamiltonian $H_t$. Letting $\rho(t)$ be the solution to
\begin{align*}
\frac{d}{dt}\rho(t)=\cS_t(\rho(t)),\quad \rho(0)=
\rho,
\end{align*}
for some initial state $\rho$ one can show that $\rho(t)=\mathcal{T}_t(\rho)$ with 
\begin{align*}
\mathcal{T}_t(\rho)=\lim\limits_{n\to\infty} \bigcirc_{l=1}^n e^{\frac{t}{n} \cS_{t_l} }(\rho).
\end{align*}
Our strategy to estimate the convergence of the relative entropy $\cT_t$ will be to apply the previous bounds for each $e^{\frac{t}{n} \cS_{t_l} }$ individually and then take the $n\to\infty$ limit. We then obtain
\begin{theorem}\label{thm:continuumlimit}
Let $\cL:\M_{2^n}\to\M_{2^n}$ be a Lindbladian with fixed point $\sigma$ satisfying a MLSI with constant $\alpha>0$. Moreover, let $\cH_t:\M_{2^n}\to\M_{2^n}$ be given by $\cH_t(X)=i[H_t,X]$ for some time-dependent Hamiltonian $H_t$. Moreover, let $\mathcal{T}_t$ be the evolution of the system under the Lindbladian $\cS_{t}=\cL+\cH_t$ from time $0$ to $t$. Then for all states $\rho$ and times $t>0$:
\begin{align}\label{equ:ourintegral}
&D(\mathcal{T}_t(\rho)||\sigma)\leq e^{-\alpha t}D(\rho||\sigma)\nonumber\\&+\int\limits_{0}^t d\tau e^{-\alpha(t-\tau)}\|\sigma^{-\frac{1}{2}}[H_t,\sigma]\sigma^{-\frac{1}{2}}\|.
\end{align}
\end{theorem}
\begin{proof}
We have
\begin{align*}
\mathcal{T}_t(\rho)=\lim\limits_{n\to\infty} \bigcirc_{l=1}^n e^{\frac{t}{n} \cS_{t_l} }(\rho)
\end{align*}
with $t_l=\frac{tl}{n}$.
It follows from Cor.~\ref{cor:manytimesteps} for any fixed $n$,time $t\geq 0$ and initial state $\rho$ we have:
\begin{align}\label{equ:inductionclaim}
&D( \bigcirc_{l=1}^n e^{\frac{l}{n} \cL_{t_l} }(\rho)||\sigma)\leq e^{-\alpha t}D(\rho||\sigma)\nonumber\\&+\alpha^{-1}\sum\limits_{l=1}^ne^{-\alpha \frac{nt-tl}{n}}(1-e^{-\frac{\alpha }{n}})\|\sigma^{-\frac{1}{2}}[H_{t_l},\sigma]\sigma^{-\frac{1}{2}}\|.
\end{align}
We will now take the limit $n\to\infty$ for a fixed $t$.
First note that:
\begin{align*}
&\alpha^{-1}\sum\limits_{l=1}^{n}e^{-\alpha \frac{nt-tl}{n}}(1-e^{-\frac{\alpha }{n}})\|\sigma^{-\frac{1}{2}}[H_{t_l},\sigma]\sigma^{-\frac{1}{2}}\|=\\&
\sum\limits_{l=1}^n \frac{1}{n}e^{-\alpha \frac{nt-tl}{n}}\|\sigma^{-\frac{1}{2}}[H_{t_l},\sigma]\sigma^{-\frac{1}{2}}\|+\cO\left(\frac{1}{n^2}\right)
\end{align*}
by performing a Taylor expansion of the $(1-e^{-\frac{\alpha }{n}})$ term.
We conclude that 
\begin{align*}
&\lim\limits_{n\to\infty}\alpha^{-1}\sum\limits_{l=1}^{n}e^{-\alpha \frac{nt-tl}{n}}(1-e^{-\frac{\alpha }{n}})\|\sigma^{-\frac{1}{2}}[H_{t_l},\sigma]\sigma^{-\frac{1}{2}}\|=\\
&\lim_{n\to\infty}\sum\limits_{l=1}^n\frac{1}{n}e^{-\alpha \frac{nt-tl}{n}}\|\sigma^{-\frac{1}{2}}[H_{t_l},\sigma]\sigma^{-\frac{1}{2}}\|=\\&\int\limits_{0}^t d\tau e^{-\alpha(t-\tau)}\|\sigma^{-\frac{1}{2}}[H_\tau,\sigma]\sigma^{-\frac{1}{2}}\|.
\end{align*}
This yields the claim.
\end{proof}
As Lemma~\ref{lem:dataprocessedtriangle} was the main building block to prove this bound, it should come as no surprise that Theorem~\ref{thm:continuumlimit} performs well in the regimes in which Lemma~\ref{lem:dataprocessedtriangle} does too. That is, either for $\sigma$ a high temperature Gibbs state or $[H_t,\sigma]=0$. In the first case, the commutator $\|\sigma^{-\frac{1}{2}}[H_{t_l},\sigma]\sigma^{-\frac{1}{2}}\|$ is small for all times. In the second case we, have that the evolution will leave $\sigma$ approximately invariant at the end of the computation, resulting in small values  $\|\sigma^{-\frac{1}{2}}[H_{t_l},\sigma]\sigma^{-\frac{1}{2}}\|$ irrespective of the temperature. On the other hand, the contribution of $\|\sigma^{-\frac{1}{2}}[H_{t_l},\sigma]\sigma^{-\frac{1}{2}}\|$ at the beginning of the computation is exponentially suppressed in the integral in Eq.~\eqref{equ:ourintegral}, even for $\sigma$ a Gibbs state at very high inverse temperature.
This phenomenon is illustrated in the main text with annealers supposed to prepare the ground state of a classical Hamiltonian and $\sigma$ a diagonal quantum state.

\subsection{Noise model for adiabatic quantum computation}\label{sec:noisemodel}
To the best of our knowledge, spontaneous emissions, decoherence and control errors are the main sources of noise in superconducting qubits~\cite{headquarters2019technical}, the leading platform for state of the art implementations of quantum annealers. We will now demonstrate our technique on a Lindbladian reflecting these sources of noise.
In order not to overcomplicate the presentation, we will assume that there is no cross-talk between the qubits and that each term of the Lindbladian only acts on one qubit. Moreover, we will model the control error as follows. At each time $s$, instead of implementing the Hamiltonian $H_s$ we implement the random Hamiltonian 
\begin{align*}
    &\tilde{H}_s=\sum_{i\sim j} (a_{i,j}+\xi_{i,j}^{z,z}(s))Z_i Z_j\\&+\sum_{i}(b_i+\xi_{i}^{z}(s))Z_i- \sum(\Gamma_i(s)+\xi_{i}^{x}(s)) X_i,
\end{align*}
where $\xi_{i}^{z},\xi_{i}^{x}$ is white noise with the same standard deviation $r_3$ and mean $0$. But it is straightforward to adapt our techniques to other noise models. Moreover, as we will see later, the $\xi_{i,j}^{z,z}$ control errors will only lead to extra dephasing that does not influence our analysis. Thus, we will set the $Z_iZ_j$ control errors to $0$.
All these assumptions lead to a Lindbladian with $\cL$ given by $\cL=\sum_i\cL_i$, where $\cL_i$ acts on qubit $i$ and has the terms:
\begin{align}\label{equ:noisemodel}
\cL_i=r_1\cL_{amp}+r_2\cL_{deph}+r_3\cL_{cont},
\end{align}
where $\cL_{amp}$ is the amplitude damping term, $\cL_{deph}$ the dephasing and $\cL_{cont}$ the control error. 
One can show that under our  white noise assumption, the control error is given by:
\begin{align*}
\cL_{cont}(\rho)=X\rho X+Z\rho Z-2\rho.
\end{align*}
We refer to e.g.~\cite{Budini_2001} for a derivation of this term.
Each local term $\cL_i$ then acts as:
\begin{align*}
&\begin{pmatrix}
a & b \\
c & d 
\end{pmatrix}\mapsto\\&\begin{pmatrix}
r_1d+r_3(d-a) & -(3r_3+r_2+\tfrac{r_1}{2})b+r_3c \\
-(3r_3+r_2+\tfrac{r_1}{2})c+r_3b& -r_1d+r_3(a-d)
\end{pmatrix}
\end{align*}
on $2\times 2$ matrices.
Moreover, we assume for simplicity that the noise rates $r_1,r_2$ and $r_3$ are the same on each qubit.
One can then easily solve $\cL_i(\sigma)=0$ for a state $\sigma$ and conclude that the fixed point of this evolution is given by
\begin{align}\label{equ:fixedstate2}
\sigma=\begin{pmatrix}
\frac{r_1+r_3}{r_1+2r_3}& 0 \\
0 &\frac{r_3}{r_1+2r_3}
\end{pmatrix}.
\end{align}
Thus, we see that in the regime $r_1\gg r_3$, the fixed point is essentially the $\ket{0}$ state, for $r_3\gg r_1$ we have a state close to maximally mixed state and for $r_1\simeq r_3$ we have the state $\tfrac{2}{3}\ketbra{0}{0}+\tfrac{1}{3}\ketbra{1}{1}$.
Note that we can write $\cL_i$ as $\cL_i=(r_1+2r_3) \cL_{\sigma}+\cL'$, where $\cL_{\sigma}(\rho)=\tr\lb \rho\rb\sigma-\rho$ and $\cL'$ is a Lindbladian s.t. $\cL'(\sigma)=0$. This implies that $\cL_i$ satisfies a MLSI with the same constant as  $(r_1+2r_3)\cL_{\sigma}.$
This gives a MLSI for $\cL$ with $\alpha(r_1,r_3)=r_1+2r_3$~\cite[Theorem 19]{Beigi_2020}. Note that the constant $r_2$ does not influence the fixed point or the relative entropy convergence. This is why we have not added the extra dephasing generated by the control errors of the $Z_iZ_j$ terms. This also indicates that our convergence analysis is suboptimal for small times. Indeed, as we pick our initial state as $\ket{\psi_0}=\ket{+}^{\otimes n}$, it is not difficult to see that the initial convergence rate scales like $r_1+2r_3+r_2$. However, at the end of the evolution we expect the state to mostly be diagonal and, thus, the action of the dephasing noise is negligible. We conclude that for large times the convergence rate is indeed of the order $r_1+2r_3$. It should also be noted that some works suggest that the dephasing noise actually occurs in the eigenbasis of $H_s$~\cite{Albash_2015}. As argued in this work, if the state at time $s$ is close to the ground state of $H_s$, the effect of dephasing in the eigenbasis of $H_s$ is negligible.  On the other hand, note that the dephasing obtained because of the $Z_iZ_j$ terms is necessarily in the computational basis.

\subsection{Computations required for linear path}\label{sec:computations}

Here we gather the computations required to exemplify how inequality~\eqref{equ:entroadiabatic} behaves for the path $H_s=\frac{s}{T}H_0+\left(1-\frac{t}{T}\right)H_I$ if we  evolve from the initial state $\rho=\ketbra{+}{+}^{\otimes n}$ under $\cS_s$ from $0$ to $T$. We will assume that the noise models is given as in Eq.~\eqref{equ:noisemodel} with fixed point $\sigma$ given by the tensor product of the state in Eq.~\eqref{equ:fixedstate}. 
Note that as $H_I$ commutes with $\sigma$, we obtain from Eq.~\eqref{equ:entroadiabatic} that:
\begin{align}\label{equ:entroadiabatic2}
&D(\mathcal{T}_T(\rho)||\sigma)\leq\nonumber\\& e^{-\alpha T}D(\rho||\sigma)+\int\limits_{0}^T d\tau e^{-\alpha(T-\tau)}\lb1-\frac{\tau}{T}\rb\|\sigma^{-\frac{1}{2}}[H_0,\sigma]\sigma^{-\frac{1}{2}}\|.
\end{align}
Integrating Eq.~\eqref{equ:entroadiabatic2}, we obtain:
\begin{align}\label{equ:boundinterpolation}
&D(\mathcal{T}_T(\rho)||\sigma)\leq e^{-\alpha T}D(\rho||\sigma)\nonumber\\&+\frac{1}{\alpha^2T}\lb 1-e^{-\alpha T}\alpha T-e^{-\alpha T}\rb\|\sigma^{-\frac{1}{2}}[H_0,\sigma]\sigma^{-\frac{1}{2}}\|.
\end{align}
By a triangle inequality and noting that the state $\sigma$ commutes with the $Z_i$ terms, we can further bound the commutant by:
\begin{align*}
\|\sigma^{-\frac{1}{2}}[H_0,\sigma]\sigma^{-\frac{1}{2}}\|\leq \sum_i |\Gamma_i|\|\sigma^{-\frac{1}{2}}[X_i,\sigma]\sigma^{-\frac{1}{2}}\|.
\end{align*}
It will be convenient to reparametrize the state as $\sigma=(e^{\gamma}+e^{-\gamma})^{-n}\otimes_{i=1}^n e^{\gamma Z_i}$ for  $\gamma=\frac{1}{2}\log\lb \frac{r_1+r_3}{r_3}\rb$ and let $r=r_1+2r_3$. 
One can then show that $\|\sigma^{-\frac{1}{2}}[X_i,\sigma]\sigma^{-\frac{1}{2}}\|=2\sinh(\gamma)$.
Putting all these elements together, we conclude that with this noise model the relative entropy after evolving the system for time $T$ is bounded by
\begin{align}\label{equ:boundentropyising2}
&D(\mathcal{T}_T(\ketbra{+}{+}^{\otimes n})||\sigma)\leq e^{-r T}n\log(2\cosh(\gamma))\nonumber\\& +\frac{2\sinh(\gamma)\lb 1-e^{-rT}r T-e^{-rT}\rb}{r^2T} n\bar{\Gamma}.
\end{align}
For the sake of conciseness, denote the r.h.s. of Eq.~\eqref{equ:boundentropyising2} by $n f(\gamma,r,T,\bar{\Gamma})$ with
\begin{align}\label{equ:locadep}
&f(\gamma,r,T,\bar{\Gamma})= e^{-r T}\log(2\cosh(\gamma))\nonumber\\& +\frac{2\sinh(\gamma)\lb 1-e^{-rT}r T-e^{-rT}\rb}{r^2T} \bar{\Gamma}.
\end{align}
In terms of the noise rates and simplifying the expressions a bit we have:
\begin{align*}
&f(r_1,r_3,T,\bar{\Gamma})\leq e^{-(r_1+2r_3) T}\log\lb \frac{r_1+2r_3}{\sqrt{r_3(r_1+r_3)}}\rb\nonumber\\& +r_1\frac{ 1-e^{-(r_1+2r_3)T}(r_1+2r_3) T-e^{-(r_1+2r_3)T}}{\sqrt{r_1+r_3}(r_1+2r_3)^2T} \bar{\Gamma}.
\end{align*}

\end{document}